\documentclass[conference]{IEEEtran}
\IEEEoverridecommandlockouts

\pdfoutput=1

\usepackage{xspace}
\usepackage{amsmath,amssymb,amsthm}
\usepackage{bm}
\usepackage{makecell}
\usepackage{multirow}
\usepackage{thmtools,thm-restate}
\usepackage{hyperref}

\usepackage[vlined,boxed,linesnumbered]{algorithm2e}
\usepackage{zref-clever}

\usepackage[normalem]{ulem}

\usepackage{enumitem}

\newcommand{\ie}{i.e.,\xspace}
\newcommand{\eg}{e.g.,\xspace}
\newcommand{\etc}{etc.\xspace}
\newcommand{\etal}{et al.\xspace}

\renewcommand{\paragraph}[1]{\smallskip\noindent \textbf{#1} \;}

\newtheorem{theorem}{Theorem}
\newtheorem{lemma}{Lemma}
\newtheorem{corollary}{Corollary}
\newtheorem{observation}{Observation}
\newtheorem{definition}{Definition}

\usepackage{mathtools}
\DeclarePairedDelimiter{\ceil}{\lceil}{\rceil}
\DeclarePairedDelimiter{\floor}{\lfloor}{\rfloor}

\DeclareMathOperator{\Forall}{\forall\,}
\DeclareMathOperator{\Exists}{\exists\,}
\DeclareMathOperator{\Nexists}{\nexists\,}

\DontPrintSemicolon

\SetKwHangingKw{InstParam}{instantiation parameter:}{}{}
\SetKwHangingKw{InstParams}{instantiation parameters:}{}{}
\SetKwHangingKw{Init}{initialization:}{}{}
\SetKwHangingKw{LocVar}{local variable of every process}{}{}
\SetKwProg{ProcCode}{code of every process}{ is}{}
\SetKwProg{UniProcCode}{code of process}{ is}{}
\SetKwIF{If}{ElseIf}{Else}{if}{then}{else if}{else}{}

\SetKwComment{Comment}{$\rhd$}{}
\SetCommentSty{textit}

\SetKwFor{Upon}{upon}{do}{}
\SetKwFor{ForAll}{for all}{do}{}

\SetKwComment{Comment}{$\rhd$}{}
\SetCommentSty{textit}
\SetArgSty{upshape}
\SetProgSty{upshape}

\let\oldnl\nl
\newcommand{\nonl}{\renewcommand{\nl}{\let\nl\oldnl}}

\usepackage{tikz}
\usetikzlibrary{
    decorations.pathreplacing,
    math, calligraphy, 
    arrows, arrows.meta,
    shapes, shapes.geometric
}

\zcsetup{
cap = true,
}
\zcRefTypeSetup{algocf}{
Name-sg = Algorithm ,
name-sg = algorithm ,
Name-sg-ab = Alg. ,
name-sg-ab = alg. ,
Name-pl = Algorithms ,
name-pl = algorithms ,
}
\zcRefTypeSetup{algocfline}{
Name-sg = line ,
name-sg = line ,
name-sg-ab = l. ,
Name-pl = lines ,
name-pl = lines ,
name-pl-ab = l. ,
}
\zcRefTypeSetup{observation}{
Name-sg = Observation ,
name-sg = observation ,
name-sg-ab = obs. ,
Name-pl = Observations ,
name-pl = observations ,
}
\zcRefTypeSetup{corollary}{
Name-sg = Corollary ,
name-sg = corollary ,
name-sg-ab = cor. ,
Name-pl = Corollaries ,
name-pl = corollaries ,
}
\zcRefTypeSetup{theorem}{
Name-sg = Theorem ,
name-sg = theorem ,
Name-sg-ab = Thm. ,
name-sg-ab = thm. ,
Name-pl = Theorems ,
name-pl = theorems ,
}
\AddToHook{env/lemma/begin}{\zcsetup{countertype={theorem=lemma}}}
\AddToHook{env/definition/begin}{\zcsetup{countertype={theorem=definition}}}
\AddToHook{env/corollary/begin}{\zcsetup{countertype={theorem=corollary}}}

\newcommand{\outputi}{\textsc{output}\xspace}

\newcommand{\outputseti}{\textsc{outputSet}\xspace}
\newcommand{\movei}{\textsc{move}\xspace}
\newcommand{\choicei}{\textsc{choice}\xspace}

\newcommand{\OutputGraph}{\ensuremath{\mathit{SOSGraph}}\xspace}

\newcommand{\CompatibilityWalk}{\ensuremath{\mathit{SOSWalk}}\xspace}

\newcommand{\leadOutSets}{\ensuremath{\mathit{LeadersOutputSets}}\xspace}

\newcommand{\wait}{\ensuremath{\mathsf{wait}}\xspace}

\newcommand{\communicate}{\ensuremath{\mathsf{communicate}}\xspace}
\newcommand{\observe}{\ensuremath{\mathsf{observe}}\xspace}
\newcommand{\observes}{\ensuremath{\mathsf{observes}}\xspace}
\newcommand{\observed}{\ensuremath{\mathsf{observed}}\xspace}

\newcommand{\ooutput}{\ensuremath{\mathsf{output}}\xspace}

\newcommand{\obs}{\ensuremath{\mathsf{obs}}}

\newcommand{\leader}{\ensuremath{\mathit{leader}\xspace}}

\newcommand{\ValUniv}{\ensuremath{\mathbb{V}}\xspace}

\newcommand{\vecpairsuniv}{\ensuremath{\mathbb{U}}\xspace}

\newcommand{\Vin}{\ensuremath{V_\text{in}}\xspace}
\newcommand{\vin}{\ensuremath{v_\text{in}}\xspace}
\newcommand{\Vout}{\ensuremath{V_\text{out}}\xspace}
\newcommand{\vout}{\ensuremath{v_\text{out}}\xspace}

\newcommand{\task}{\ensuremath{T}\xspace}
\newcommand{\taskn}{\ensuremath{T_n}\xspace}

\newcommand{\Vvecin}{\ensuremath{\mathcal{I}}\xspace}
\newcommand{\Vvecout}{\ensuremath{\mathcal{O}}\xspace}
\newcommand{\map}{\ensuremath{\Delta}\xspace}

\newcommand{\Vvecinn}{\ensuremath{\mathcal{I}_n}\xspace}
\newcommand{\Vvecoutn}{\ensuremath{\mathcal{O}_n}\xspace}
\newcommand{\mapn}{\ensuremath{\Delta_n}\xspace}

\newcommand{\SV}{\ensuremath{\textit{OS}}\xspace}
\newcommand{\ST}{\ensuremath{\textit{SOS}}\xspace}

\newcommand{\information}{\ensuremath{\textit{I}}\xspace}

\newcommand{\idx}{\ensuremath{\textit{i}}\xspace}
\newcommand{\iin}{\ensuremath{\textit{in}}\xspace}
\newcommand{\out}{\ensuremath{\textit{out}}\xspace}

\newcommand{\valence}{\ensuremath{\mathit{Val}}\xspace}
\newcommand{\sstate}{\ensuremath{\sigma}\xspace}
\newcommand{\States}{\ensuremath{\Sigma}\xspace}
\newcommand{\instates}{\ensuremath{\mathit{InStates}}\xspace}

\newcommand{\Asynchrony}{\ensuremath{\mathit{Asynchrony}}\xspace}
\newcommand{\Resilience}{\ensuremath{\mathit{Resilience}}\xspace}

\newcommand{\Termination}{\ensuremath{\mathit{Termination}}\xspace}

\newcommand{\pb}{\ensuremath{\Pi}\xspace}

\newif\ifannote

\ifannote
    \newcommand{\anninsert}[2]{{\color{#1}#2}}
    \newcommand{\anncomment}[3]{
        {\color{#1}\colorbox{#1}{\bfseries\sffamily\tiny\textcolor{white}{#2}}
        $\blacktriangleright$ \em #3 $\blacktriangleleft$}
    }
\else
    \newcommand{\anninsert}[2]{#2}
    \newcommand{\anncomment}[3]{}
\fi

\newcommand{\TA}[1]{\anncomment{red}{TA}{#1}}

\newcommand{\ta}[1]{\anninsert{red}{#1}}
\newcommand{\af}[1]{\anninsert{violet}{#1}}

\newcommand{\jw}[1]{\anninsert{brown}{#1}}

\newcommand{\xOmit}[1]{}
\newcommand{\del}[1]{}

\title{%
On the Decidability of Distributed Tasks with Output Sets under Asynchrony and Any Number of Crashes
}

\author{\IEEEauthorblockN{Timoth\'e Albouy}
\IEEEauthorblockA{\textit{IMDEA Software Institute}\\
Madrid, Spain
}
\and
\IEEEauthorblockN{Antonio Fern\'andez Anta}
\IEEEauthorblockA{\textit{IMDEA Software \& IMDEA Networks Institutes} \\
Madrid, Spain
}
\and
\IEEEauthorblockN{Chryssis Georgiou}
\IEEEauthorblockA{\textit{University of Cyprus}\\
Nicosia, Cyprus
}
\and
\IEEEauthorblockN{Nicolas Nicolaou}
\IEEEauthorblockA{\textit{Algolysis Ltd}\\
Nicosia, Cyprus
}
\and
\IEEEauthorblockN{Junlang Wang}
\IEEEauthorblockA{\textit{IMDEA Software Institute \& Universidad Carlos III de Madrid} \\
Madrid, Spain
}
}

\begin{document}

\maketitle

\begin{abstract}
\ta{This paper studies the decidability of \textit{task problems}, i.e., distributed problems expressed as sets of distributed tasks. Specifically, we introduce a new class of task problems called \textit{Set of Output Sets} (\textit{SOS}) \textit{problems}. An SOS problem $\pb_O$ is defined by a set $O$ (called \textit{SOS}), and requires that the set of sets of distinct output values produced across all executions corresponds exactly to $O$.} We then demonstrate that this class of \ta{problems} is decidable: there is a procedure determining whether any SOS \ta{problem} is solvable asynchronously under \jw{$f$} crashes. The decision rule is as follows. Every SOS \ta{problem} is solvable when \ta{$f=0$}. For \ta{$f > 0$}, an SOS \ta{problem} is solvable if and only if the graph $G=(O,\subset)$ is connected. In this graph, each vertex is an output set in $O$, and two vertices are linked by an edge whenever one output set includes the other. One of the surprising implications of our results is that, replacing validity by a \textit{completeness} property (which guarantees that all output sets of size at most $k$ are produced), $k$-set agreement is solvable under any number of crashes \ta{$f \geq 0$} for $k>1$, and unsolvable under \ta{$f>0$} crashes only for $k=1$ (consensus). Finally, we study a novel family of \ta{problems} called $d$-disagreement, which requires the system to always produce $d$ different output values, and we show that its implementability condition is related to the harmonic series.
\end{abstract}








\TA{
\begin{itemize}
\item cite full version \cite{AFGNW26} (wait for instructions)
\end{itemize}
}

\section{Introduction}

Many problems of distributed computing can be represented as \ta{sets of} \textit{tasks}.
\ta{A task is an abstraction} where every participant has at most one input value and at most one output value.
\ta{For example, the consensus problem contains all tasks where} each participant proposes a value (the inputs), and all non-faulty participants decide on the same value (the outputs) from the proposals.
Although tasks cannot capture some distributed problems (in particular, long-lived primitives such as shared registers or transactional objects), they remain a particularly useful abstraction for studying the solvability boundaries of distributed computing.
\ta{This paper focuses on \textit{task problems}, \ie distributed problems expressed as sets of tasks.}

\paragraph{Towards a theory of distributed decidability.}
One of the grand challenges of distributed computing involves determining which families of distributed problems are decidable and which ones are not.
Decidability in this context means that there exists an effective procedure determining if a given distributed problem is solvable in a given computing model or not.
It has been shown that distributed computing as a whole is undecidable, as undecidable \ta{classes of problems} have been identified within it, \ta{such as} the \ta{class of asynchronous task problems under} at least 2 crashes~\cite{GK99,HR97}.
But undecidability in the general case does not prevent some classes of \ta{problems} to still be decidable: indeed, all asynchronous \ta{task problems} under at most 1 crash~\cite{BD89,BMZ90,MW87}, or in the presence of only initial crashes~\cite{TKM89}, are decidable.
The present work is part of this endeavour: it focuses on a particular class of \ta{task problems} that we call \ta{\emph{SOS problems}}, and it shows that there exists a procedure determining whether any of these SOS \ta{problems} can be implemented asynchronously under any number of crashes.
Eventually, by progressively identifying new classes of decidable and undecidable \ta{distributed problems}, it will be possible to build an exhaustive theory of distributed decidability that is as fundamental to this field as computability theory is to sequential computing.

\ta{
\paragraph{The limits of the classical task formalism.}
Combinatorial topology is one of the most powerful approaches for studying the solvability and decidability of task problems~\cite{HKR13}.
In this approach, a distributed task is a triplet $T=(\Vvecin,\Vvecout,\map)$, where the input and output domains $\Vvecin,\Vvecout$ are modelled as \textit{simplicial complexes}, and the input-output mapping $\map$ is a \textit{carrier map} (a.k.a. task specification) between the input and output simplexes.
With this framework, tasks are divided in two categories: \textit{colorless} tasks, where any process can adopt the input or output value of any other process~\cite{R16}, and \textit{colored} tasks, which do not have this restriction.
Combinatorial topology yielded several far-reaching results, such as BG simulation~\cite{BGLR01,G09} or the ACT theorem~\cite{HS99}.

However, the simplicial machinery imposes an important restriction on the definition of tasks (whether colorless or colored), namely, that the input and output domains of a task are ``closed under inclusion'': any subsimplex of a legal simplex is also legal.
In practice, this means that many important distributed problems, such as strong symmetry breaking (SSB)~\cite{AP16} or election~\cite{IRR11,SP89}, cannot be effectively captured as simplicial complexes: they admit executions outputting several distinct values (\eg both $0$ and $1$), but not executions outputting each value individually.
Capturing specifications of this kind is precisely the goal of the present work.



}

\paragraph{Our approach: focusing on the output sets of tasks.}
In this paper, we partially address the \ta{expressivity} limitations of traditional approaches by \ta{developping a framework for defining and studying tasks that relaxes the closedness constraint of simplicial approaches.}


Building upon this framework, we extend the \textit{Set of Output Sets} (\textit{SOS}) approach introduced by Albouy \etal for tasks with binary output values~\cite{AFGNW25}.
Specifically, we formalize and define the class of \textit{SOS \ta{problems}}, capturing both binary and multivalued output sets, where validity, output multiplicity, and process identities are disregarded.
These SOS \ta{problems} are defined by the set of distinct output values that they can produce across their executions, called their \textit{set of output sets}.
Our approach does not capture \ta{validity- or process-ID-dependent task problems for instance}, but in exchange, it allows us to identify a new class of distributed \af{decidable problems} via the study of SOS \ta{problems} (whereas it has been shown that the classical approach is undecidable~\cite{GK99,HR97}).

\paragraph{Contributions and roadmap.}
Our contributions are the following.
\begin{itemize}[leftmargin=*]
    \item \textit{A novel framework for studying distributed tasks.}
    We introduce in \zcref{sec:model-pb} a new methodology for expressing and analyzing tasks.
    \ta{
    Specifically, we present a formalization of tasks that relaxes the closedness constraint of simplex-based definitions.
    Moreover, unlike the traditional approach which characterizes decidability at the task level (a task being an instance of a distributed problem), we characterize it at the level of \textit{task problems}, which are sets of tasks with some common properties.}
    
    \item \textit{A new class of distributed decidability: SOS \ta{problems}.}
    We study in \zcref{sec:decidable-sos} the particular class of SOS \ta{problems $\pb_O$}, which are defined by the set of output sets $O$ that they can produce.
    We show that this class is decidable: we present a decision rule determining whether any SOS \ta{problem $\pb_O$} is solvable asynchronously under some number of crashes \jw{$f$}.
    This rule can be stated as: \ta{$\pb_O$} is always solvable if \jw{$f=0$}, and it is solvable with \jw{$f>0$} if and only if the graph $G=(O,\subset)$ is weakly connected (\zcref{sec:algorithm-1,sec:non-resil}).
    The vertices of this graph are the output sets in $O$, and its edges are derived from the inclusions between these output sets.\footnote{
        Specifically, we consider the undirected version of $G$, which is sometimes referred to as the \emph{comparability graph} of $(O,\subset)$.
    }
    Interestingly, our results have direct implications for $k$-set agreement.
    Indeed, we can define \textit{complete $k$-set agreement} (a \ta{weaker} version of $k$-set agreement with a completeness property instead of validity) as an SOS \ta{problem}, and show that it can be solved under more than $k$ crashes for $k>1$, but only without crashes for $k=1$, i.e, consensus (\zcref{sec:ConseusVSkSA}).
    This demonstrates a fundamental cut between consensus and ${>}1$-set agreement.
    To the best of our knowledge, we are the first ones to study the decidability of a particular subclass of \ta{task problems} under any crash tolerance, instead of studying the decidability of all \ta{task problems} under a specific crash tolerance.
    
    \item \textit{A particular family of SOS \ta{problems}: $d$-disagreement.}
    Finally, in \zcref{sec:disagreement} we look at a special case of SOS \ta{problems}, which we call $d$-disagreement, and which requires the system to always produce $d$ different output values (hence the system ``disagrees''). 
    This problem can have multiple interesting applications, especially in the context of fault-tolerant load balancing~\cite{GS08}.
    For example, if a distributed system has to execute $d$ idempotent jobs $J_1,...,J_d$ in parallel in the presence of up to \jw{$f$} crashes, $d$-disagreement can be used, such that every process with output value $i \in [1..d]$ executes job $J_i$.
    This effectively guarantees that all jobs are eventually executed despite crashes. 
    %
    The $d$-disagreement problem was introduced in its binary version by Albouy \etal~\cite{AFGNW25}, but the present paper provides its definition in terms of SOS \ta{problems} and generalizes it to the multivalued case.
    We then provide an impossibility proof and an algorithm demonstrating that the implementability condition of $d$-disagreement is tied to the harmonic series.
    Indeed, under \jw{$f$} crashes, the number of processes required to solve disagreement is approximately \jw{$H_d(f+1)$}, where $H_d$ is the $d$-th harmonic number.
\end{itemize}

\zcref{sec:related-work} exposes the research landscape in which this work is situated. Our computing model is presented in \zcref{sec:model-pb}, and we discuss the 
resilience and decidability of SOS \ta{problems} in \zcref{sec:decidable-sos} to \zcref{sec:non-resil}.
We do show how decidability results relate to $k$-set agreement and consensus in \zcref{sec:ConseusVSkSA}.
Lastly we present $d$-disagreement, a newly discovered family of symmetry breaking problem, in \zcref{sec:disagreement}.
Concluding remarks are provided in \zcref{sec:conclusion}. 

\section{Related Work} \label{sec:related-work}

\paragraph{Combinatorial topology for distributed computing.}
As discussed earlier, combinatorial topology has found surprising applications in the study of distributed tasks, as exposed in the monograph by Herlihy, Kozlov, and Rajsbaum~\cite{HKR13}.
This approach finds its roots in STOC 1993, where three concurrent papers leveraged combinatorial or topological results~\cite{BG93,HS93,SZ93} (journal versions:~\cite{BGLR01,HS99,SZ00}), such as Sperner's lemma or Brouwer's fixed point theorem, to prove the impossibility of $k$-set agreement under asynchrony and $k$ crashes~\cite{C93}.
Specifically, the last two papers~\cite{HS93,SZ93} (which won the 2004 Gödel prize) proved this impossibility in the wait-free case, while the first one~\cite{BG93} (which won the 2017 Dijkstra prize) introduces the \textit{BG simulation} technique, extending the impossibility to the general case.
Broadly speaking, \ta{the original} BG simulation shows that adding more correct processes to a wait-free system does not increase its computability power, for a specific class of tasks called \textit{colorless} tasks~\cite{BGLR01} (also called convergence tasks).
\ta{
Later results studied the \textit{colored} case, either by restricting to 3-process tasks~\cite{AFPR25}, or via the extended BG simulation~\cite{G09}, characterizing $f$-resilient solvability for the full class of classical (downward-closed) \textit{colored} tasks, but not for specifications outside that model, such as strong symmetry breaking (SSB)~\cite{IRR11}, election~\cite{IRR11,SP89}, or the $d$-disagreement problem we introduce here.
}



\paragraph{Decidability of distributed \ta{problems}.}
We focus on the asynchronous crash-prone model, as most of the literature on distributed decidability studied this particular model.
Moran and Wolfstahl generalized the famous impossibility of consensus~\cite{FLP85} by proving that all tasks that have a connected \textit{input graph} and disconnected \textit{output graph} do not tolerate crashes~\cite{MW87} (informally, the input (resp. output) graph is constructed by linking the input (resp. output) vectors that differ by only one value).
This result was later completed by Biran, Moran, and Zaks, who demonstrated the decidability of \ta{task problems} in the 1-resilient message-passing model (shown to be equivalent to 1-resilient shared memory~\cite{BD89}), notably by introducing a general algorithm solving all 1-resilient tasks in their classification~\cite{BMZ90}.
\af{Note that the works in \cite{BMZ90,MW87} are restricted to the \jw{$f=1$} case.}
\ta{We use similar graph-theoretic techniques, but we study the general \jw{$f \leq n$} case.}
In the same vein, Taubenfeld, Katz, and Moran showed the decidability of \ta{task problems} in the presence of initial failures~\cite{TKM89}.
The first \textit{undecidability} result of distributed computing is due to Gafni and Koutsoupias, who, by drawing connections with the contractability problem in topology, showed that some 3-process tasks are undecidable in the shared-memory wait-free model (with 2 crashes), i.e, it is impossible to say if these tasks are solvable or not~\cite{GK99}.
This undecidability result was extended to other communication models and arbitrary levels of resilience by Herlihy and Rajsbaum~\cite{HR97}.
The known decidability results of distributed \ta{task problems} in the asynchronous crash-prone model are summarized in \zcref{tab:current-classification}.

\begin{table}[b]
\centering
\footnotesize
\begin{tabular}{|c|c|c|c|}
\hline
\textbf{Family} & \textbf{Condition on crashes} & \textbf{Decidable} & \textbf{Reference} \\
\hline\hline
\multirow{3}{*}{\makecell{Arbitrary \ta{task} \\ \ta{problems}}} 
  & \jw{$f \geq 2$} & No & \cite{GK99,HR97} \\
\cline{2-4}
  & \jw{$f=1$} & Yes & \cite{BD89,BMZ90,MW87} \\
\cline{2-4}
  & initial crashes & Yes & \cite{TKM89} \\
\hline
\makecell{Binary SOS \\ \ta{problems} ($*$)} & any \jw{$f \leq n$} & Yes & \cite{AFGNW25} \\
\hline
SOS \ta{problems} & any \jw{$f \leq n$} & Yes & \textbf{This work} \\
\hline
\end{tabular}
\caption{Decidability results on distributed \ta{problems} identified to date (under asynchrony and crashes).\vspace{-2em}
}
\label{tab:current-classification}
\end{table}

\paragraph{The SOS approach.}
The present work is closest to that of Albouy \etal~\cite{AFGNW25}, who introduced the Set of Output Sets approach in the binary case (\ie the output values are $0$ or $1$).
However, they did not formally define the concept of SOS \ta{problems}.
In contrast, our formalization of \textit{SOS \ta{problems}} can capture both binary and multivalued output sets, and it allows us to transfer the tightness results of~\cite{AFGNW25} to a precise class of \ta{problems}, noted ``Binary SOS \ta{problems} ($*$)'' in \zcref{tab:current-classification}.

\section{Computing Model and Problem Formalization} \label{sec:model-pb}

For clarity, we provide in \zcref{tab:notations} a list of concepts and notations used in this paper.

\begin{table}[tb]
\centering
\footnotesize
\begin{tabular}{|c|c|}
    \hline
    \textbf{Concept or notation} & \textbf{Meaning} \\
    \hline\hline
    $p \in P$ &
    Process $p$ in the set of processes $P$
    \\
    \hline
    $n=|P|$ & Number of processes in the system
    \\
    \hline
    \jw{$f$} & Maximum number of crashed processes
    \\
    \hline
    $\pb, \task, A, E$ & \ta{Task problem, task,} algorithm, execution \\
    \hline
    $\ValUniv$ & Universe of values \\
    \hline
    $\vecpairsuniv$ & Universe of pairs of input/output vectors \\
    \hline
    $\bot \notin \ValUniv$ & Sentinel value denoting no input/output
    \\
    \hline
    $\Vin,\Vout \in (\ValUniv {\cup} \{\bot\})^n$
    & Input and output vectors of size $n$ \\
    \hline
    $\star$ & Unspecified value \\
    \hline
\end{tabular}
\caption{Concepts and notations used in this paper.}\vspace{-2em}
\label{tab:notations}
\end{table}

\subsection{Computing Model} \label{sec:model}

\paragraph{Process model.}
We consider a distributed system with a non-empty set $P=\{p,p',...\}$ of processes ($|P| = n$).
Processes are deterministic computing entities that take steps according to their local state and the events they observe, following an algorithm~$A$.
Failures are restricted to crash faults: in an execution of the system, a process may crash (\ie halt prematurely and take no further steps), but it does not deviate from its algorithm before crashing.
We assume an upper bound \jw{$f$} on the number of processes that may crash in an execution, with \jw{$0 \leq f \leq n$}.
A process that does not crash in an execution $E$ is said to be \emph{correct} in~$E$.
We also assume that all processes have access to their own private local clocks, which may run at a different pace, and which allow them to set local timeouts.

\paragraph{Communication model.}
Processes interact through a generic asynchronous communication medium that is \textit{reliable}: it does not suppress, duplicate, or corrupt communicated information.
Every process $p \in P$ has access to two abstract operations:
\begin{itemize}[leftmargin=*]
    \item $\communicate$ $\information$: process $p$ disseminates some information $\information$ to all the system processes;
    \item $\observe$ $\information$ (callback event): process $p$ is notified that information $\information$ was communicated.
\end{itemize}
The medium guarantees that all correct processes eventually obtain a consistent view of the set of communicated information, while crashed processes may only have partial but always valid views.
More formally, the communication abstraction satisfies the following properties (the ``C'' prefix stands for ``communication'').
\begin{itemize}[leftmargin=*]
    \item \textbf{C-Validity:} If a process $p$ observes information $\information$, then $\information$ must have been previously communicated by some process~$p'$.
    \item \textbf{C-Integrity:} Any process $p$ observes some information $\information$ communicated by a given process~$p'$ at most once.
    \item \textbf{C-Local-Termination:} If a correct process $p$ communicates information $\information$, then some correct process~$p'$ (if there is any) eventually observes~$\information$.
    \item \textbf{C-Global-Termination:} If a correct process $p$ observes information $\information$, then all correct processes eventually observe~$\information$.
\end{itemize}
%
These communicate/observe operations can be implemented straightforwardly in asynchronous communication media such as message-passing networks (\eg using flooding) or shared memory (\eg using register scanning), regardless of the number of crashes~\ta{$f$}~\cite{HT93}.
We emphasize that our model does not place any bounds on communication delays.
Therefore, any ordering or timing of observations that is consistent with the above properties can occur.
This assumption is standard in asynchronous models and is only used in some of our correctness proofs (namely, \zcref{thm:connected-sos,thm:non-resil}).

\subsection{Problem Formalization} \label{sec:pb-formal}

\paragraph{Input and output vectors.}
For some \ta{$n \in \mathbb{N}^{*}$}, a given instance of a distributed task takes an input vector $\Vin = (\vin^1, ..., \vin^n)$ and returns an output vector $\Vout=(\vout^1, ..., \vout^n)$, where the $i$-th value $\vin^i$ in \Vin (resp. $\vout^i$ in \Vout) corresponds to the input (resp. output) value of process $p_i \in P$.
The values in the input and output vectors belong to the set $\ValUniv \cup \{\bot\}$, where $\ValUniv$ is the universe of values (of finite size) and $\bot \notin \ValUniv$ is a sentinel value that represents no input or no output.
In an execution, a process $p \in P$ can output at most one value $v$ using the operation $\ooutput$ $v$ (multiple invocations of $\ooutput$ by $p$ must be for the same value $v$).

Let $\vecpairsuniv = \bigcup_{n \in \mathbb{N}^{*}} ((\ValUniv \cup \{\bot\})^n) \times (\ValUniv \cup \{\bot\})^n)$ be the universe of all pairs of input/output vectors: the $n$-exponentiation gives the set of all possible vectors of size $n$ and the union merges all pairs of different sizes obtained previously.

\paragraph{\ta{Tasks, task problems}, and algorithms.}
\ta{A \textit{distributed task} is a triplet $\task=(\Vvecin,\Vvecout,\map)$, where $\Vvecin$ and $\Vvecout$ are sets of input/output vectors, respectively, and $\map \subseteq \Vvecin \times \Vvecout$ is a nonempty relation between input and output vectors} representing all admissible input-output mappings of \ta{the task}.
\ta{Two vectors in a pair $(\Vin,\Vout) \in \map$ must be of the same size $n$.
However, different pairs in $\map$ can have vectors of different sizes, i.e.,
we allow a task $\task=(\Vvecin,\Vvecout,\map)$ to be defined for multiple sizes $n$.
}

A \ta{\textit{task problem} $\pb=\{\task_1,\task_2,...\}$ is a set of tasks with some common properties\footnote{
  To illustrate the notion of common properties, the consensus problem contains all tasks $(\Vvecin,\Vvecout,\map)$ where, for every pair $(\Vin,\Vout) \in \map$, all non-$\bot$ values in $\Vout$ are the same $v$, and $v \in \Vin$.
  We give another example of common properties for a task problem in \zcref{def:sos-pb} (SOS problem).
}.
All the tasks in a task problem can be seen as instances of this problem.
For example, the consensus problem contains majority consensus tasks, median consensus tasks, etc.
}
A \ta{task problem $\pb$} is \emph{solvable} under asynchrony and \jw{$f$} crashes if at least one of its \ta{tasks $\task \in \pb$} can be implemented by some asynchronous algorithm tolerating up to \jw{$f$} crashes.
\ta{Let $A$ be an asynchronous algorithm tolerating up to \jw{$f$} crashes that admits a set $\Vvecin_A$ of input vectors and $\Vvecout_A$ of output vectors, and that produces a set $\map_A \subseteq \Vvecin_A \times \Vvecout_A$ of $(\Vin,\Vout)$ pairs across all its executions.
Informally, we say that algorithm $A$ \emph{implements} a task~$\task$ if and only if $\task=(\Vvecin_A,\Vvecout_A,\map_A)$.
}
A \ta{task $\task=(\Vvecin,\Vvecout,\map)$} \emph{supports} a system size $n$ if \ta{$\map$} contains at least one pair of vectors of size $n$; likewise, an algorithm $A$ \emph{supports} a system size $n$ if it admits executions with $n$ processes.
A class of \ta{task problems $\mathcal{P} = \{\pb_1, \pb_2, \ldots\}$} is \emph{decidable} under asynchrony and up to \jw{$f$} crashes if there exists an effective procedure determining whether each \ta{problem in~$\mathcal{P}$} is solvable.

\paragraph{Set of output sets (SOS).}
We define the \emph{output set} of an output vector $\Vout$ as the set of distinct output values in $\Vout$, ignoring order, multiplicity, and~$\bot$: $\SV(\Vout) = \{\vout^i \in \Vout \mid \vout^i \neq \bot\}$.
The \emph{set of output sets} (SOS) of a \ta{task $\task$ is then $\ST(\task) = \{\SV(\Vout) \mid (\Vin,\Vout) \in \task\}$}.
Every SOS $O \subseteq 2^\ValUniv$ must be non-empty, but $O$ can contain the empty set if the task instance has executions that produce no output values.
As $\ValUniv$ is finite, then $O$ is also necessarily finite.

\paragraph{SOS \ta{problems}.}
Given some set of output sets $O$, an \ta{\textit{SOS problem} $\pb_O$ is a task problem that}, for every defined system size~$n$, contains all \ta{tasks} that \textit{(i)}~produce SOS~$O$ and \textit{(ii)}~have a full mapping between input and output vectors of the same size (\ie inputs are not used to determine outputs, hence the task is ``validity-less'').
More formally, we have the following.


\begin{definition}[SOS \ta{problem}] \label{def:sos-pb}
An SOS \ta{problem $\pb_O$} with SOS $O \subseteq 2^\ValUniv, O \neq \varnothing$ is defined as follows.
Let \ta{$\taskn=(\Vvecinn,\Vvecoutn,\mapn)$ be the restriction of a task $\task=(\Vvecin,\Vvecout,\map)$ containing only vectors of size $n$}.
\ta{
\begin{align*}
    \pb_O = \big\{\task \subseteq \vecpairsuniv, \task \neq \varnothing \,&\bigl\vert\, 
\Forall n \in \mathbb{N}: \taskn \neq \varnothing \implies \\
&\big(\underbrace{\ST(\taskn) = O}_\text{\textup{constraint} \textit{(i)}}
    \land\, \underbrace{\mapn = \Vvecinn \times \Vvecoutn}_\text{\textup{constraint} \textit{(ii)}} \big)\big\}.
\end{align*}
}
\end{definition}





For example, consider a \ta{task $\task=(\Vvecin,\Vvecout,\map)$} that must produce $O = \{\{1\},\{1,2\}\}$ for pairs of vectors of size $n = 2$ given $\ValUniv = \{1,2\}$.
Constraint \textit{(i)} requires that the output sets \ta{corresponding to $\Vvecout_2$} are exactly $O$.
Concretely, \ta{$\Vvecout_2$} must contain some output vectors in $\{(1,1),(1,\bot),(\bot,1)\}$ to yield output set $\{1\} \in O$, and some output vectors in $\{(1,2),(2,1)\}$ to yield output set $\{1,2\} \in O$.
Constraint \textit{(ii)} requires that every input vector in \ta{$\map_2$} (\eg $(1,1),(1,2),(2,\bot),...$) is paired with each of these output vectors (\ie inputs do not constrain outputs).
In other words, \ta{$\map_2$} is the Cartesian product of all size-$2$ input vectors with all size-$2$ output vectors.

Therefore, solving an SOS \ta{problem $\pb_O$} under asynchrony and \ta{$f$} crashes reduces to providing an asynchronous algorithm $A$ tolerating up to \ta{$f$} crashes and satisfying constraints \textit{(i)} and \textit{(ii)} of \zcref{def:sos-pb}.
For constraint \textit{(ii)}, it suffices that $A$ ignores the input values of the processes: for a given size $n$, if all output vectors can be produced regardless of the input vector, then there is a total mapping between inputs and outputs.
For constraint \textit{(i)}, $A$ must fulfill the following two properties.
%
\begin{itemize}[leftmargin=*]
    \item \emph{Safety}: Every execution of $A$ produces an output vector $\Vout$ such that $\SV(\Vout) \in O$.
    \item \emph{Completeness}: For every system size $n$ supported by $A$ and every output set $o \in O$, there is an execution of $A$ producing an output vector $\Vout$ of size $n$ such that $\SV(\Vout) = o$.
\end{itemize}

Observe that, without the completeness property, safety can be straightforwardly achieved by algorithms where processes always output the same values (provided sufficient correct processes).
Hence, completeness acts as a form of non-triviality or weak validity in our approach.

Notice that the above specification does not impose \textit{global termination}, \ie that all non-faulty processes of the system eventually output a value.\footnote{
    Let us however note that \zcref{alg:async-d-dis-2dt,alg:async-non-resil} already guarantee global termination, but \zcref{alg:async-resil} does not.
}
However, as our definition of the solvability of SOS \ta{problems} is quite weak, the impossibility results that we give also apply to stronger formulations, such as those including global termination.



\section{SOS \ta{Problems} are Decidable} \label{sec:decidable-sos}

This section proves our main theorem (\zcref{thm:sos-decidable}), stating that the entire class of SOS \ta{problems} is decidable under asynchrony and any number of crashes \ta{$f$}.
That is, we can determine whether any SOS \ta{problem $\pb_O$} is solvable asynchronously with up to \jw{$f$} crashes or not.
Our SOS \ta{problem} classification relies on the concept of SOS graphs, defined as follows.

\begin{definition}[SOS graph] \label{def:output-graph}
\sloppy{Given an SOS \ta{problem $\pb_O$} with an SOS $O$, the $\OutputGraph(O)$ function returns the \emph{SOS graph} $(O,E)$ of \ta{$\pb_O$}, where the set of edges is $E \triangleq \{\{o,o'\} \mid o \subset o' \lor o' \subset o \}$.}
\end{definition}

Informally, the SOS graph of an \ta{SOS problem $\pb_O$} links every pair of compatible output sets $o,o' \in O$, in the sense that it would not break safety if some processes ``think'' that the system produces $o$, while some other ``think'' that the system produces $o'$.
Note that the SOS graph of an SOS $O$ is the undirected version of graph $(O,\subset)$.
For convenience, we say that an SOS $O$ (or SOS problem $\pb_O$) is \textit{connected} if its SOS graph $\OutputGraph(O)$ is connected, and \textit{disconnected} otherwise.
\zcref{fig:sos-graphs} gives two examples of connected and disconnected SOS graphs.
The left-hand SOS graph corresponds to the SOS \ta{problem $\pb_O$} with SOS $O=\{\{1\},\{3\},\{1,2\},\{1,3\},\{2,3\}\}$.
Similarly, the right-hand SOS graph corresponds to the SOS \ta{problem $\pb_{O'}$} with SOS $O'=\{\{1\},\{1,2\},\{1,3\},\{2,3\}\}$.
We now state our main theorem.

\begin{figure}[tb]
\centering
\begin{tikzpicture}[scale=0.8, transform shape]

\tikzmath{
\dx = .7;
\yTop = 1;
}

\node[ellipse, draw] (l12) at (0, \yTop) {$1,2$};
\node[ellipse, draw] (l1) at (\dx, 0) {$1$};
\node[ellipse, draw] (l13) at (2*\dx, \yTop) {$1,3$};
\node[ellipse, draw] (l3) at (3*\dx, 0) {$3$};
\node[ellipse, draw] (l23) at (4*\dx, \yTop) {$2,3$};

\draw (l12) -- (l1);
\draw (l1) -- (l13);
\draw (l13) -- (l3);
\draw (l3) -- (l23);

\tikzset{shift={(6, 0)}}

\node[ellipse, draw] (r12) at (0, \yTop) {$1,2$};
\node[ellipse, draw] (r1) at (\dx, 0) {$1$};
\node[ellipse, draw] (r13) at (2*\dx, \yTop) {$1,3$};
\node[ellipse, draw] (r23) at (4*\dx, \yTop) {$2,3$};

\draw (r12) -- (r1);
\draw (r1) -- (r13);

\end{tikzpicture}
\caption{Examples of a connected SOS graph (on the left) for SOS $\{\{1\},\{3\},\{1,2\},\{1,3\},\{2,3\}\}$ and a disconnected SOS graph (on the right) for SOS $\{\{1\},\{1,2\},\{1,3\},\{2,3\}\}$.
}
\label{fig:sos-graphs}
\end{figure}
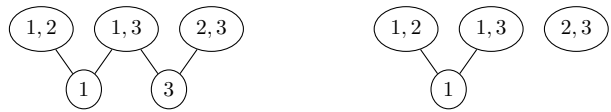

\begin{theorem}[Decidability of SOS \ta{problems}] \label{thm:sos-decidable}
An \ta{SOS problem $\pb_O$} with set of output sets $O$ can be solved asynchronously under up to \jw{$f \in \mathbb{N}$} crashes if and only if \jw{$f=0$} or $\OutputGraph(O)$ is connected.
\end{theorem}

Therefore, SOS \ta{problems} can be partitioned into two subclasses: \textit{(1)} the SOS \ta{problems} that can tolerate any number of crashes (their SOS graph is connected), and \textit{(2)} the SOS \ta{problems} that do not tolerate any crash (their SOS graph is disconnected).
The proof of the above theorem follows from \zcref{alg:async-resil}/\zcref{thm:connected-sos}, presented in the next section and addressing case \textit{(1)} for connected SOS \ta{problems}, and from \zcref{thm:disconnected} and \zcref{alg:async-non-resil}/\zcref{thm:non-resil}, presented in \zcref{sec:non-resil} and addressing case \textit{(2)} for disconnected SOS \ta{problems}.
(The results presented in \zcref{sec:non-resil} are quite close to those of \cite{MW87}, however, the translation between the two models is not straightforward, hence we provide our own proofs in the appendix for completeness.)


Notice that the SOS \ta{problem $\pb_O$} with SOS $O=\{\varnothing\}$ can be trivially solved under any number of crashes, simply by having all processes of the system do nothing.
Therefore, for the sake of simplicity, we consider in our classification that the empty graph $G=(\varnothing,\varnothing)$ is a special case of a connected graph, but in the remainder of this section, we only consider SOS \ta{problems $\pb_O$} whose SOS $O$ is different from $\{\varnothing\}$.

\section{Asynchronous \texorpdfstring{$f$}{f}-Resilient Algorithm for any Connected SOS}
\label{sec:algorithm-1}

\begin{algorithm}[tb]
\footnotesize
\InstParams{a finite walk $\CompatibilityWalk_O$ of $\OutputGraph(O)$.} \label{line:async-resil:instant}
\smallskip

\Init{pick a set of leader processes $P_\leader \subseteq P$ such that \jw{$|P_\leader| > f$};
given $V = \bigcup O$ the set of all possible output values in $O$, create partition $P_1,...,P_{|V|}$ of $P$ s.t. 
 \jw{$\Forall i \in [1 .. |V|], |P_i| > f$}.
} \label{line:async-resil:init}
\smallskip

\ProcCode{$p_\leader \in P_\leader$}{ \label{line:async-resil:pLeader}
    \For{$i=1$ \textbf{to} $|\CompatibilityWalk_O|$}{ \label{line:async-resil:iteratePATH}
        \lIf{$i \neq |\CompatibilityWalk_O|$}{%
            \communicate $\movei(i+1)$; \label{line:async-resil:leader-move}
        } 
        $\wait$ for an arbitrary local time; \\ \label{line:async-resil:leader-wait}
        \If{$\observed$ $\movei(i+1)$ from less than $|P_\leader|$ leaders}{ \label{line:async-resil:leader-NotEveryoneMove}
            $o_i \gets$ the $i$-th output set in $\CompatibilityWalk_O$; \\
            \communicate $\outputseti(o_{i})$; \\ \label{line:async-resil:leader-stay}
            $\mathsf{break}$.   \label{line:async-resil:leader-stop}
        }
    }
}
\smallskip

\ProcCode{$p_v \in P_v, v \in V$}{ \label{line:async-resil:outputV}
    \Upon{$p_v$ \observes some $\outputseti(o)$ where $v \in o$}{ \label{line:async-resil:outputV-observe}
        $\ooutput$ $v$; \\ \label{line:async-resil:outputV-output}
        $\mathsf{exit}$. \label{line:async-resil:outputV-stop}
        } 
}
\caption{Asynchronous algorithm implementing all connected SOS \ta{problems $\pb_O$} with SOS $O$ assuming \jw{$n \geq |V|(f+1)$}, where $V=\bigcup O$ is the set of all output values in $O$.}
\label{alg:async-resil}
\end{algorithm}

This section introduces \zcref{alg:async-resil}, a universal \jw{$f$}-resilient algorithm that can produce any SOS \ta{problem $\pb_O$} with a connected SOS $O$ under any number of crashes \jw{$f \geq 0$}, provided that \jw{$n \geq (f+1)|V|$}, where $V=\bigcup O$ is the set of values in $O$.
This algorithm is not tight in the sense that the proportion \jw{$f/n$} of crashes tolerated in the system is usually lower than what is optimal, but it shows the solvability of all connected SOS \ta{problems} under any number of crashes \jw{$f$}.

The algorithm is instantiated (\zcref{line:async-resil:instant}) by providing as a parameter $\CompatibilityWalk_O$, a finite walk\footnote{
    In graph theory, a walk is similar to a path, but the same edge or vertex can be visited multiple times in a walk.
    In our case, a finite walk $\CompatibilityWalk_O$ can always be constructed as the SOS $O$ is finite by definition (see \zcref{sec:pb-formal}).
} that visits all output sets in $G=\OutputGraph(O)$.
By assumption, $G$ is a connected graph, hence we can construct a walk $\CompatibilityWalk_O$ visiting all vertices of $G$.
$\CompatibilityWalk_O=(o_1,o_2,...)$ is encoded as a sequence of output sets that can be iterated in order.

%

At the initialization of the algorithm~(\zcref{line:async-resil:init}), a set of more than \ta{$f$} leader processes $P_\leader \subseteq P$ is selected, ensuring that at least one leader process does not crash.
We then create a partition of $|V|$ distinct exhaustive subsets of $P$ (where $V$ is the set of all output values in $O$), each of size greater than \ta{$f$} and associated with a distinct value $v \in V$.
We only assume \jw{$n \geq |V|(f+1)$}, hence every leader process $p_\leader \in P_\leader$ is also contained in some subset $P_v$ of the partition of $P$.

Each leader process $p_\leader \in P_\leader$ iterates $i$ from $1$ to $|\CompatibilityWalk_O|$~(\zcref{line:async-resil:iteratePATH}).
If $i$ is smaller than $|\CompatibilityWalk_O|$, $p_\leader$ communicates $\movei(i+1)$~(\zcref{line:async-resil:leader-move}).
After waiting an arbitrary local time~(\zcref{line:async-resil:leader-wait}), $p_\leader$ checks whether it observed $\movei(i+1)$ from all other leaders~(\zcref{line:async-resil:leader-NotEveryoneMove}).
If it is not the case, $p_\leader$ communicates $\outputseti(o_i)$ with the $i$-th element in $\CompatibilityWalk_O$ (\zcref{line:async-resil:leader-stay}) and then exits the loop~(\zcref{line:async-resil:leader-stop}).
When a process $p_v \in P_v$ (where $v \in V$) observes some $\outputseti(o)$ information such that $v$ is in $o$, $p_v$ outputs $v$ (\zcref{line:async-resil:outputV-output}) and exits (\zcref{line:async-resil:outputV-stop}).

\paragraph{Correctness of~\zcref{alg:async-resil}.} \label{sec:univ-alg-correct}
Now we prove \zcref{thm:connected-sos} which states the correctness of \zcref{alg:async-resil}.
We first show some preliminary results.
In the following, we denote by $\leadOutSets$ the set containing all output sets $o$ from the $\outputseti(o)$ communicated by the group of leader processes $P_\leader$ during an execution of \zcref{alg:async-non-resil}.

\begin{lemma} \label{lem:async-resil:leader-always-output}
$\leadOutSets$ is nonempty.
\end{lemma}

\begin{proof}
As \jw{$|P_\leader| > f$} by construction, there is at least one non-faulty leader process $p_\leader \in P_\leader$.
Recall that $\CompatibilityWalk_O$ is necessarily of finite size.
%
For some loop index $i \in [1..|\CompatibilityWalk_O|]$, $p_\leader$ eventually enters the condition at \zcref{line:async-resil:leader-move} for one of the following reasons:
\begin{itemize}[leftmargin=*]
    \item the communication was delayed due to asynchrony, and $p_\leader$ did not observe all $\movei(i+1)$ by the end of the \wait at \zcref{line:async-resil:leader-wait};

    \item some leader process $p'_\leader \in P_\leader$ did not communicate $\movei(i+1)$, either because $p'_\leader$ crashed, or because $o_i$ was the last element of $\CompatibilityWalk_O$ at \zcref{line:async-resil:leader-move}.
\end{itemize}

%
Hence, the process $p_\leader$ eventually communicates some $\outputseti(o_i)$ at \zcref{line:async-resil:leader-stay}, and therefore $\leadOutSets$ is nonempty.
\end{proof}

\begin{lemma} \label{lem:async-resil:leader-traversal}
For every output set $o_i \in \CompatibilityWalk_O$, there is an execution of \zcref{alg:async-resil} where $\leadOutSets = \{o_i\}$.
\end{lemma}

\begin{proof}
As $\CompatibilityWalk_O$ is a sequence, every element $o_i \in \CompatibilityWalk_O$ has a finite index $i \in \mathbb{N^*}$.
There must exist a crash-free execution where, for every loop index $j \in [1..i-1]$, every leader process communicates $\movei(j+1)$ at \zcref{line:async-resil:leader-move}, observes all $|P_\leader|$ $\movei(o_j+1)$ by the end of the \wait at \zcref{line:async-resil:leader-wait}, and therefore skips the condition at \zcref{line:async-resil:leader-NotEveryoneMove}.
(If $i=1$, then the loop is executed only once.)
At the end of this procedure, loop index $i$ is reached.
There must exist a crash-free execution that reaches this stage, and all leader processes do not observe all $|P_\leader|$ $\movei(i+1)$ by the end of the \wait at \zcref{line:async-resil:leader-wait} due to asynchrony (since any timing of observations consistent with the communication properties is admissible, see \zcref{sec:model}).
Thus, all leader processes enter the condition at \zcref{line:async-resil:leader-NotEveryoneMove}, communicate $\outputseti(o_i)$ at \zcref{line:async-resil:leader-stay}, and exit the loop at \zcref{line:async-resil:leader-stop}.
Hence, $\leadOutSets = \{o_i\}$.
\end{proof}

\begin{observation} \label{obs:leader-comm-outputset}
By construction (\zcref{line:async-resil:iteratePATH,line:async-resil:leader-stay}), all elements of $\leadOutSets$ are in $\CompatibilityWalk_O$.
\end{observation}

\begin{lemma} \label{lem:async-resil:1-or-2-output-sets}
$\leadOutSets$ is either of the form $\{o_i\}$ or of the form $\{o_i,o_{i+1}\}$, where $o_i,o_{i+1}$ are two subsequent elements in $\CompatibilityWalk_O$.
\end{lemma}

\begin{proof}
\zcref{lem:async-resil:leader-always-output} shows that $\leadOutSets$ cannot be empty.
\zcref{lem:async-resil:leader-traversal} implies the existence of some executions where $\leadOutSets$ is a singleton $\{o_i\}$ where $o_i \in \CompatibilityWalk_O$.

Let us consider an execution where $\leadOutSets$ is at least of size $2$.
By \zcref{obs:leader-comm-outputset}, $\leadOutSets$ contains elements of $\CompatibilityWalk_O$, and therefore we can consider any two distinct output sets $o_i,o_j \in \leadOutSets$, where $i$ and $j$ are the indices of $o_i,o_j$ in $\CompatibilityWalk_O$ such that $i<j$.
By definition of $\leadOutSets$, there is one leader process $p_\leader \in P_\leader$ that communicated some $\outputseti(o_i)$ at \zcref{line:async-resil:leader-stay}, during the loop iteration $i$, and then exited the loop at \zcref{line:async-resil:leader-stop}.
Therefore, $p_\leader$ could not have communicated $\movei(i+2)$ during the loop iteration $i+1$, and every leader process that reaches the loop iteration $i+1$ thus enters the condition at \zcref{line:async-resil:leader-NotEveryoneMove}, communicates $\outputseti(o_{i+1})$ at \zcref{line:async-resil:leader-stay} and exits the loop at \zcref{line:async-resil:leader-stop}.
It follows that $j=i+1$, and thus that $\leadOutSets$ is of the form $\{o_i,o_{i+1}\}$.
\end{proof}

\begin{lemma} \label{lem:async-resil:leader-gives-output-set}
In an execution with set $\leadOutSets$, the set of output values produced by the execution is $o=\bigcup \leadOutSets$.
\end{lemma}

\begin{proof}
Consider an execution with set $\leadOutSets$, let $o=\bigcup \leadOutSets$ be the set of all values appearing in $\leadOutSets$, and let $V=\bigcup O$ be the set of all values appearing in $O$.

By definition, $\leadOutSets$ is the set of all output sets $o'$ that have been communicated by leader processes in some $\outputseti(o')$ at \zcref{line:async-resil:leader-stay}.
By C-Local-Termination and C-Global-Termination, all correct processes eventually observe all of these $\outputseti(o')$.
Since the union of all $o' \in \leadOutSets$ is $o$, and by \zcref{obs:leader-comm-outputset}, we deduce that $o \subseteq V$.
As \jw{$|P_v|>f$} for every $v \in V$, there is at least one correct process $p_v \in P_v$ that outputs $v$ at \zcref{line:async-resil:outputV} upon observing the first $\outputseti(o')$ such that $v \in o'$.
Therefore, as a whole, the execution produces output set~$o$.
\end{proof}





\begin{theorem}[Correctness of \zcref{alg:async-resil}]
\label{thm:connected-sos}
\zcref{alg:async-resil} \ta{solves} any given connected SOS \ta{problem $\pb_O$} with SOS $O$ under asynchrony and up to \jw{$f$} crash failures, assuming \jw{$n \geq |V|(f+1)$}, where $V = \bigcup O$ is the set of all possible output value in $O$.
\end{theorem}


\begin{proof}
\ta{We prove safety and completeness of \zcref{alg:async-resil}.}

\sloppy
\paragraph{Safety.}
Let us consider an arbitrary execution of \zcref{alg:async-resil} with set $\leadOutSets$.
By \zcref{lem:async-resil:1-or-2-output-sets}, there are only two cases: \textit{(i)} $\leadOutSets=\{o_i\}$, or \textit{(ii)} $\leadOutSets=\{o_i,o_{i+1}\}$, where $o_i,o_{i+1}$ are two subsequent elements in $\CompatibilityWalk_O$.
\zcref{lem:async-resil:leader-gives-output-set} shows that the output set produced by the execution is $o=\bigcup \leadOutSets$.
For case \textit{(i)}, the output set of the execution is $o=\bigcup \{o_i\}=o_i$, which is a valid output set of $O$.
For case \textit{(ii)}, as $o_i,o_{i+1}$ are subsequent elements in $\CompatibilityWalk_O$, then they are linked by an edge in the SOS graph $\OutputGraph(O)$.
By \zcref{def:output-graph}, one of these two sets must include the other.
\ta{W.l.o.g.}, let us assume that $o_i \subset o_{i+1}$.
Thus, the output set of the execution is $o=\bigcup \{o_i,o_{i+1}\}=o_{i+1}$, which is also a valid output set of $O$.
Therefore, any execution of \zcref{alg:async-resil} produces a valid output set.



\paragraph{Completeness.}
By \zcref{lem:async-resil:leader-traversal}, for every output set $o \in O$, there is an execution in which $\leadOutSets = \{o_i\}$.
By \zcref{lem:async-resil:leader-gives-output-set}, the output set produced by this execution is $o_i$.
%
\end{proof}





\section{The Case of Non-Resilient SOS \ta{Problems}} \label{sec:non-resil}

We provide in this section the remaining results of our decidability study of SOS \ta{problems}, namely: \textit{(i)} that no disconnected SOS \ta{problem} can be solved under asynchrony and crashes, \textit{(ii)} that all SOS \ta{problems} can be solved asynchronously without crashes with \zcref{alg:async-non-resil}, and \textit{(iii)} that $n \geq \max\{|o| : o \in O\}$ and \jw{$f=0$} are tight conditions to \ta{solve} disconnected SOS \ta{problems} asynchronously.

The first result is summarized in the following theorem, whose proof can be found in \zcref{sec:disconnected}.

\begin{theorem}[Non-resilience of disconnected SOS \ta{problems}] \label{thm:disconnected}
A disconnected SOS \ta{problem $\pb_O$} can be \ta{solved} asynchronously only if \jw{$f=0$}.
\end{theorem}


We now present \zcref{alg:async-non-resil}, a universal algorithm that \ta{solves} any SOS \ta{problem $\pb_O$} under asynchrony assuming $n \geq \max\{|o| : o \in O\}$, $n \geq 1$, and \jw{$f=0$} (no crash).
The algorithm is instantiated by providing the SOS $O$ as a parameter (\zcref{line:async-non-resil:instant}).
At the initialization of the algorithm (\zcref{line:async-non-resil:init}), a leader process $p_\leader \in P$ is selected.
We also let $m$ be the size of the biggest output set in $O$, and create a partition of all processes $P_1,...,P_m$ containing $m$ nonempty subsets (we can create this partition as we assume $n \geq \max\{|o| : o \in O\}=m$).
Since $P$ is nonempty, \ta{the existence of a leader is guaranteed}.

First, for every output sets $o \in O$, the leader process $p_\leader$ communicates $\choicei(o)$ at \zcref{line:async-non-resil:iteratePATH}.
Then, $p_\leader$ waits for observing the first $\choicei(o')$ at \zcref{line:async-non-resil:pLeader-observe}, and finally, $p_\leader$ communicates $\outputseti(o')$  at \zcref{line:async-non-resil:pLeader-comm}.
Besides, every process \ta{$p \in P_i$} (including $p_\leader$) first waits for the $\outputseti(o)$ communicated by $p_\leader$ (\zcref{line:async-non-resil:wait}).
Then, \ta{$p$} computes its output value $v_i \in o$ based on its subset index $i$ in the partition ($1 \leq i \leq m$): more precisely, \ta{$p$} sorts all values in $o$ and takes the $((i \bmod |o|)+1)$-th one in this sequence (\zcref{line:async-non-resil:choose-v}).
Finally, $p_i$ outputs $v_i$ (\zcref{line:async-non-resil:outputV-output}).

\begin{algorithm}[b]
\footnotesize
\InstParam{SOS $O$.} \label{line:async-non-resil:instant}
\smallskip

\Init{pick some leader process $p_\leader \in P$, let $m=\max\{|o| : o \in O\}$, create a partition $P_1,\cdots,P_m$ of $P$ s.t. $\Forall i \in [1..m], |P_i| \geq 1 $.
} \label{line:async-non-resil:init}
\smallskip

\UniProcCode{$p_\leader$}{ \label{line:async-non-resil:pLeader}
    \lForAll{$o \in O$}{%
        \communicate $\choicei(o)$; \label{line:async-non-resil:iteratePATH}
    }
    \wait $p_\leader$ \observes first $\choicei(o')$; \\ \label{line:async-non-resil:pLeader-observe}
    \communicate $\outputseti(o')$. \label{line:async-non-resil:pLeader-comm}
}
\smallskip

\ProcCode{\ta{$p \in P_i, i \in [1..m]$}}{ \label{line:async-non-resil:outputV}
    \ta{\wait $p$} \observes $\outputseti(o)$; \label{line:async-non-resil:wait} \\
    $v_i \gets ((i \mod |o|)+1)$-th largest value in $o$; \label{line:async-non-resil:choose-v} \\
    $\ooutput$ $v_i$. \label{line:async-non-resil:outputV-output}
}
\caption{Asynchronous algorithm implementing all possible \ta{SOS problems $\pb_O$ with} SOS $O$ assuming $n \geq \max\{|o| : o \in O\}$, $n \geq 1$, and \jw{$f=0$}.}
\label{alg:async-non-resil}
\end{algorithm}

The correctness of \zcref{alg:async-non-resil} is given in \zcref{thm:non-resil}.
We begin with an intermediary lemma.

\begin{lemma} \label{lem:comm-is-os}
If $p_\leader$ communicates $\outputseti(o)$ at \zcref{line:async-non-resil:pLeader-comm}, then the execution produces output set $o$.
\end{lemma}

\begin{proof}
Let us assume that process $p_\leader$ communicates $\outputseti(o)$ at \zcref{line:async-non-resil:pLeader-comm}.
By C-Local-Termination, C-Global-Termination, and the fact that there are no crashes, all processes eventually observe $\outputseti(o)$ and pass the \wait statement at \zcref{line:async-non-resil:wait}.
By C-Validity and the fact that only $p_\leader$ communicates information in the algorithm, $\outputseti(o)$ is the only information observed by the processes.
We can first observe that the execution cannot produce output values that are not in $o$, by \zcref{line:async-non-resil:outputV}. 

Since all subsets in the partition $P_1,...,P_m$ are nonempty and no process crashes, all processes \ta{$p^1 \in P_1$} pick the smallest value $v_1 \in o$, all processes \ta{$p^2 \in P_2$} pick the second smallest value $v_2$, \etc, and finally, all processes \ta{$p^{|o|} \in P_{|o|}$} pick the biggest value $v_{|o|} \in o$ (recall that $|o| \leq m$).
After that, all processes \ta{$p \in P_i,i \in [1..|o|]$} output $v_i$ at \zcref{line:async-non-resil:outputV}. 
Therefore, as a whole, all processes of $P_1,...,P_{|o|}$ output all values in $o$, and the execution produces the output set $o$.
\end{proof}

\begin{theorem} \label{thm:non-resil}
\zcref{alg:async-non-resil} \ta{solves} any SOS \ta{problem $\pb_O$} with SOS $O$ under asynchrony, assuming $n \geq \max\{|o| : o \in O\}$, and \jw{$f=0$}.  
\end{theorem}

\begin{proof}
We prove 
safety and completeness of \zcref{alg:async-non-resil}. 

\paragraph{Safety.}
Let $o$ be an output set produced by an execution of \zcref{alg:async-non-resil}, we will now show that $o \in O$.
Since there is no crash, $p_\leader$ has observed at \zcref{line:async-non-resil:pLeader-observe} some $\choicei(o')$ that it has communicated before (by C-Validity) and therefore $o' \in O$ (from \zcref{line:async-non-resil:iteratePATH}).
After that, $p_\leader$ communicated $\outputseti(o')$ at \zcref{line:async-non-resil:pLeader-comm}.
By \zcref{lem:comm-is-os}, the execution produces output set $o'$, and $o=o'$.
As $o' \in O$, then $o \in O$.


\paragraph{Completeness.}
For every $o \in O$, there must exist an execution where the first information observed by the leader process $p_\leader$ is $\choicei(o)$ due to asynchrony (since any timing of observations consistent with the communication properties \af{must occur}, see \zcref{sec:model}).
By \zcref{lem:comm-is-os}, the execution must produce the output set $o$.
\end{proof}


Combining the previous results we show that the two conditions $n \geq \max\{|o| : o \in O\}$ and \jw{$f=0$} are necessary and sufficient to implement any \textit{disconnected} SOS \ta{problem $\pb_O$} in asynchrony.
In the following, we consider an arbitrary SOS \ta{problem $\pb_O$} with a disconnected SOS $O$, \ie $\OutputGraph(O)$ is a disconnected graph.

\paragraph{Necessity of $\bm{n \geq \max\{|o| : o \in O\}}$ and \jw{$\bm{f=0}$}.}
The necessity of \ta{$f=0$} comes from \zcref{thm:disconnected}.
The necessity of $n \geq \max\{|o| : o \in O\}$ comes from the fact that, if this condition is not satisfied, then there are not enough processes in the system to output all the values of the largest output set $o \in O$, violating completeness.



\paragraph{Sufficiency of $\bm{n \geq \max\{|o| : o \in O\}}$ and \jw{$\bm{f=0}$}.}
If $\OutputGraph(O)$ is a disconnected graph, then it necessarily contains several vertices, and SOS $O$ cannot be $\{\varnothing\}$.
It implies that $n \geq 1$.
Finally, if we assume that conditions $n \geq \max\{|o| : o \in O\}$ and \jw{$f=0$} are satisfied, then we have all the sufficient conditions to use \zcref{alg:async-non-resil} and implement the disconnected SOS \ta{problem $\pb_O$}.

\section{Consensus versus \texorpdfstring{$k$}{k}-Set Agreement}
\label{sec:ConseusVSkSA}

Our decidability results for SOS \ta{problems} have surprising implications, especially in the case of $k$-set agreement~\cite{C93}, a family of \ta{distributed problems where each execution outputs} at least one and at most $k \in \mathbb{N}^*$ different values among the input values.
Consensus is a special case of $k$-set agreement where $k=1$.
It has been shown that $k$-set agreement under asynchrony and up to $k$ crashes is impossible~\cite{BGLR01,HS99,SZ00}, if the validity of the task is considered (\ie decisions are proposals).
However, it is interesting to note that this validity assumption is unnecessary to prove the impossibility of $1$-set agreement/consensus~\cite{FLP85,AFGGNW24}: to preclude trivial implementations, it only suffices to require that consensus may produce different output values, without imposing that the output values are also input values.
This non-triviality property of consensus can be expressed as a \textit{completeness} restriction: for instance, non-trivial binary consensus must produce both output sets $\{0\}$ and $\{1\}$ (as described in the original FLP theorem~\cite{FLP85}).
Our work demonstrates a remarkable discontinuity in $k$-set agreement: its impossibility relies on validity only for $k \geq 2$, but not for $k=1$.
(See \zcref{sec:disconnected} for a generalized impossibility proof of the case $k=1$.)

\subsection{The case of \texorpdfstring{$k$}{k}-set agreement (\texorpdfstring{$k>1$}{k>1})} \label{sec:complete-kSA}

In the following, we define the family of \textit{complete} $k$-set agreement \ta{problems} via our \ta{SOS approach}.
Intuitively, in complete $k$-set agreement, all nonempty output sets of size at most $k$ can be produced, independently of the input values of the execution (\ie the classical $k$-set agreement validity property does not hold).
We then show that, for $k \geq 2$, these \ta{problems} can be solved under any number of crashes, as their SOS graph is necessarily connected.

\begin{definition}[Complete $k$-set agreement] \label{def:kSA-SOS}
Given a universe $\ValUniv$ of at least $k+1$ values, \ta{\emph{complete $k$-set agreement} is an SOS \ta{problem $\pb_O$} with SOS $O = \{o \subseteq \ValUniv : 1 \leq |o| \leq k\}$}.
\end{definition}

For instance, if $\ValUniv=\{1,2,3\}$, the SOS of complete $2$-set agreement is $O=\{\{1\},\{2\},\{3\},\{1,2\},\{1,3\},\{2,3\}\}$.
The SOS graph associated with $O$ is connected, so complete $2$-set agreement can be solved with any number of crashes.
More generally, the following holds.

\begin{corollary}[Solvability of complete $k$-set agreement with $k \geq 2$]
If $k \geq 2$, then complete $k$-set agreement is solvable under any number of crashes \jw{$f$}.
\end{corollary}

\begin{proof}
Since $k \geq 2$, the universe $\ValUniv$ has at least $3$ values, and $O$ contains every nonempty subset of $\ValUniv$ of size at most $k$.
Hence, every output set $o \in O$ is connected to (or is itself) a singleton $\{v\} \in O$ (property 1), and any two singletons $\{v\},\{v'\} \in O$ are connected via the path $\{v\} \subset \{v,v'\} \supset \{v'\}$, since $\{v,v'\} \in O$ whenever $k \geq 2$ (property 2).
Together, these imply that $\OutputGraph(O)$ is connected, so the result follows from \zcref{alg:async-resil} via \zcref{thm:connected-sos}.
\end{proof}

\subsection{The case of consensus (\texorpdfstring{$k=1$}{k=1})}
The impossibility of asynchronous resilient consensus follows directly from \zcref{thm:disconnected}, as consensus has a disconnected SOS.
Similarly to how we defined \textit{complete} $k$-set agreement using our SOS approach (\zcref{sec:complete-kSA}), we can formally define \textit{complete} consensus as follows.

\begin{definition}[Complete consensus] \label{def:consensus-sos}
\emph{Complete consensus} is an SOS \ta{problem $\pb_O$} with SOS $O$ such that:
\begin{enumerate}[leftmargin=*]
    \item $\Forall o \in O: |o|=1$, that is, every execution of consensus produces only one output value;

    \item $|O| \geq 2$, that is, consensus has at least two different executions that produce different output values.
\end{enumerate}
\end{definition}

The first item of \zcref{def:consensus-sos} imposes that the system always agrees on one single output value, while the second item precludes trivial implementations that always output the same value across all executions.
We can see that complete $k$-set agreement (\zcref{def:kSA-SOS}, page~\pageref{def:kSA-SOS}) boils down to complete consensus (\zcref{def:consensus-sos}) when $k=1$.

\begin{corollary}[Impossibility of asynchronous resilient consensus, FLP theorem~\cite{FLP85}] \label{cor:flp}
Consensus is impossible under asynchrony and one crash.
\end{corollary}

\begin{proof}
By \zcref{def:consensus-sos}, complete consensus \ta{$\pb_O$} has an SOS $O$ which contains at least two different singletons $\{v\},\{v'\} \in O$.
Neither of these two singletons can include the other, hence $\OutputGraph(O)$ is a disconnected graph.
Therefore, \zcref{thm:disconnected} applies, which entails that complete consensus cannot be solved asynchronously with \jw{$f>0$}.
As complete consensus (\zcref{def:consensus-sos}) is an even weaker problem than validity-based consensus (which must also guarantee that every execution produces an output value taken from the input values), the impossibility of validity-based consensus follows.
\end{proof}




\section{The \texorpdfstring{$d$}{d}-Disagreement Problem} \label{sec:disagreement}

In \zcref{sec:algorithm-1} we saw that all connected SOS \ta{problems} can be implemented using \zcref{alg:async-resil}.
However, this algorithm makes a strong system assumption, namely, that there is at least one correct process for every possible output value.
This raises a new question: Can we identify some subclasses of connected SOS \ta{problems} with better implementability conditions?
We now turn to this question for a natural family of connected SOS \ta{problems}, called $d$-disagreement, that requires every execution to produce exactly $d \geq 1$ distinct output values.
Formally, a $d$-disagreement SOS \ta{problem $\pb_O$} has SOS $O=\{\{v_1,...,v_d\}\}$.
As $d$-disagreement seeks to bound the \textit{minimum} number of different decision values, it belongs to the \textit{symmetry breaking} family of distributed problems (in contrast to \textit{agreement} problems, which seek to bound the \textit{maximum} number of different decision values).

In the following, we first present in \zcref{sec:disag-impo} an impossibility proof on the implementability of $d$-disagreement, namely, that it can be implemented under asynchrony only if \jw{$n \geq \sum_{i=1}^d \big\lceil \frac{f+1}{i} \big\rceil$}.
Notice that this bound approximates (up to the iterated rounding) to the $d$-th harmonic number times \jw{$f+1$}.
We then provide in \zcref{sec:disag-impl} an asynchronous implementation of $d$-disagreement that assumes \jw{$n \geq d\big\lceil \frac{f+1}{2} \big\rceil + \big\lfloor \frac{f+1}{2} \big\rfloor$}.

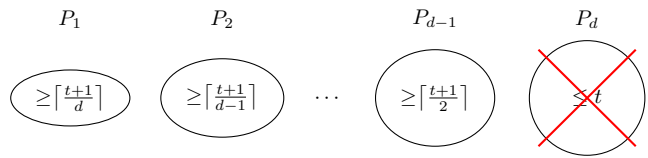
\begin{figure}[b]
\centering
\begin{tikzpicture}[scale=0.8, transform shape]

\tikzmath{
\xPi = 0;
\xPii = 2.5;
\xPdi = 6;
\xPd = 8.5;
\yLbl = 1.3;
\xEllipsis = (\xPii + \xPdi)/2;
\hwCross = .8;
\hhCross = .8;
}

\node[
    ellipse, draw=black
] at (\xPi, 0) {\jw{${\geq}\big\lceil \frac{f+1}{d} \big\rceil$}};
\node at (\xPi, \yLbl) {$P_1$};

\node[
    ellipse, draw=black, minimum height=1.3cm,
] at (\xPii, 0) {\jw{${\geq}\big\lceil \frac{f+1}{d-1} \big\rceil$}};
\node at (\xPii, \yLbl) {$P_2$};

\node at (\xEllipsis, 0) {$\cdots$};

\node[
    ellipse, draw=black,
    minimum height=1.6cm,
] at (\xPdi, 0) {\jw{${\geq}\big\lceil \frac{f+1}{2} \big\rceil$}};
\node at (\xPdi, \yLbl) {$P_{d-1}$};

\node[
    ellipse, draw=black,
    minimum width=1.9cm, minimum height=1.9cm,
] at (\xPd, 0) {\jw{$\leq f$}};
\node at (\xPd, \yLbl) {$P_d$};

\draw[red, semithick] (\xPd - \hwCross, -\hhCross) -- (\xPd + \hwCross, \hhCross);
\draw[red, semithick] (\xPd - \hwCross, \hhCross) -- (\xPd + \hwCross, -\hhCross);

\end{tikzpicture}
\caption{Execution $E$, where all processes in the set $P_d$ crash.}
\label{fig:impo-d-disag}
\end{figure}

\subsection{Impossibility of \texorpdfstring{$d$}{d}-disagreement} \label{sec:disag-impo}

\TA{Generalize this impossibility to all SOS whose minimal output set is of size $d$?}

\begin{theorem}[Necessity of $d$-disagreement] \label{thm:nec-async-d-disag}
Under asynchrony, condition 
\jw{$n \geq \sum_{i=1}^d \big\lceil \frac{f+1}{i} \big\rceil$}
is necessary for implementing the sets of output sets $O = \{\{v_1,\dots,v_d\}\}$.
\end{theorem}

\begin{proof}
For simplicity of presentation, in this proof, we assume the existence of a global clock that is inaccessible to the processes.
By contradiction, let us assume that there is an algorithm $A$ that implements the set of output sets $\{\{v_1,\dots,v_d\}\}$ under asynchrony and \jw{$n < \sum_{i=1}^d \big\lceil \frac{f+1}{i} \big\rceil$}.
We present an execution $E$ of $A$ with \jw{$f$} crashes that violates this assumption.
Execution $E$ is illustrated in \zcref{fig:impo-d-disag}, and corresponds to the concatenation of execution fragments described as follows.

Consider the first fragment of $E$, denoted $E^1$, which starts at time $0$.
Initially, we let the processes run $A$ freely until the first process $p_1^1$ is about to output some value $v_1^1$ at a time $\tau_1^1$. Process $p_1^1$ is frozen at time $\tau_1^1$, so no other process can observe $v_1^1$ (for now).
Then, we let the non-frozen processes run $A$ freely until a second process $p_1^2$ is about to output some value $v_1^2$ at a time $\tau_1^2$, which is also frozen at $\tau_1^2$ so no other process can observe $v_1^2$.
We repeat this process until process \jw{$p_1^{f+1}$} is about to output value \jw{$v_1^{f+1}$} and is frozen at time \jw{$\tau_1^{f+1}$}.
This is the time execution fragment $E^1$ ends.
(Remark that, if there are enough processes, such a process \jw{$p_1^{f+1}$} must exist, since otherwise at most \jw{$f$} processes output values in the execution; then, crashing them instead of freezing them would create an execution that \ta{cannot} produce the output set $\{v_1,\dots,v_d\}$.)
Let $w_1$ be the most frequent value in \jw{$v_1^1, \dots, v_1^{f+1}$}, and $P_1$ be the subset of processes that output $w_1$ in $E^1$.
Observe that \jw{$|P_1| \geq \big\lceil \frac{f+1}{d} \big\rceil$}.
Let us also denote $c_1=|P_1|$.

Now, consider the second fragment $E^2$ starting at time \jw{$\tau_1^{f+1}$}.
Initially, in this fragment, all processes in $P_1$ are thawed, so they can go ahead and output $w_1$.
Then, we let the non-frozen processes run freely $A$ as long as they only output value $w_1$.
We stop when a process $p_2^1$ is about to output some value $v_2^1 \neq w_1$ at a time $\tau_2^2$.
Process $p_2^1$ is then frozen at time $\tau_2^1$, so no other process can observe $v_2^1$.
We repeat this $c_1$ times until there are (again) \jw{$f+1$} frozen processes.
At this time $\tau_2^{c_1}$ is when execution fragment $E^2$ ends.
(By the same argument as before, this must always happen if there are enough processes.)
Let $w_2$ be the most frequent value to be output by the frozen processes and $P_2$ be the subset of frozen processes that output $w_2$. 
Since $w_1$ is not output by any frozen process, we have that \jw{$c_2=|P_2| \geq \big\lceil \frac{f+1}{d-1} \big\rceil$}.


Each execution fragment $E^i$, for $i=3,..., d-1$, is constructed inductively in a similar way. It starts at time $\tau_{i-1}^{c_{i-2}}$  by thawing all processes in $P_{i-1}$, so they can output $w_{i-1}$.
Then, we let non-frozen processes run $A$, while freezing the first $c_{i-1}$ processes that are about to output values not in $\{w_1, w_2, ..., w_{i-1}\}.$
When this happens (at time $\tau_i^{c_{i-1}}$) fragment $E^i$ ends with \jw{$f+1$} frozen processes.
Let $w_i$ be the most frequent value to be output by the frozen processes and $P_i$ be the subset of frozen processes that output $w_i$. 
Since $w_i \notin \{w_1, w_2, ..., w_{i-1}\}$, we have that \jw{$c_i=|P_i| \geq \big\lceil \frac{f+1}{d-i+1} \big\rceil$}.

The final execution fragment $E^d$ starts at time $\tau_{d-1}^{c_{d-2}}$ by thawing all processes in $P_{d-1}$, so they can output $w_{d-1}$.
Observe that all the remaining frozen processes are about to output the only value $w_d \in \{v_1,\dots,v_d\}$ that is not in
$\{w_1,\dots,w_{d-1}\}$.
Then, we let non-frozen processes run $A$, while freezing the first $c_{d-1} - 1$ processes that are about to output value
$w_d$.
When this happens (at time $\tau_d^{c_{d-1}}$) fragment $E^d$ ends with a set $P_d$ of \jw{$c_d=f$} frozen processes.
At that time, we crash the \jw{$f$} processes of $P_d$, so they never output value $w_d$.
This is the time execution $E$ ends, since all processes either have output a value or crashed, given that
\jw{$
\sum_{i=1}^d |P_i| = \sum_{i=1}^d \Big\lceil \frac{f+1}{i} \Big\rceil -1 \geq n.
$ 
}
By construction, the set of values $\{v_1,\dots,v_d\}$ has not been produced in execution $E$.
Hence, we have a contradiction.
\end{proof}

\subsection{Implementation of \texorpdfstring{$d$}{d}-disagreement} \label{sec:disag-impl}

In this section, we present \zcref{alg:async-d-dis-2dt}, an asynchronous implementation of the $d$-disagreement SOS \ta{problem $\pb_O$}, where $O = \{\{v_1,...,v_d\}\}$, assuming \jw{$n \geq d \big\lceil \frac{f+1}{2} \big\rceil+ \big\lfloor \frac{f+1}{2} \big\rfloor$}.
At the initialization of the algorithm, the set of all processes $P$ is partitioned into $d+1$ subsets, $P_1,..., P_d, P_?$, such that \jw{$|P_?| = n-d\big\lceil \frac{f+1}{2} \big\rceil$}, and for every $i \in [1..d]$, \jw{$|P_i| = \big\lceil \frac{f+1}{2} \big\rceil$} (\zcref{line:async-d-dis-2dt:init}).\footnote{
    Note that any construction of $P_1,...,P_d,P_?$ such that the union of two subsets is strictly greater than \jw{$f$} also works.
}
We can easily see that the assumption \jw{$n \geq d \big\lceil \frac{f+1}{2} \big\rceil + \big\lfloor \frac{f+1}{2} \big\rfloor$} is sufficient to instantiate all sets of processes $P_1,...,P_d,P_?$.
Remark that \jw{$|P_?|=n-d\big\lceil \frac{f+1}{2} \big\rceil \geq \big\lfloor \frac{t+1}{2} \big\rfloor$}.
Therefore, this construction guarantees that, in any pair of subsets $P_i,P_j$ of the partition, there is always at least one correct process.\footnote{
    Recall that, $\Forall k \in \mathbb{N}: k = \big\lfloor \frac{k}{2} \big\rfloor + \big\lceil \frac{k}{2} \big\rceil$.
}




%

\begin{algorithm}[tb]
\footnotesize
\Init{create partition $P_1,...,P_d,P_?$ of $P$ s.t. 
 \jw{$|P_?| = n-d \big\lceil \frac{f+1}{2} \big\rceil, \Forall i \in [1..d]: |P_i| = \big\lceil \frac{f+1}{2} \big\rceil$}.%
} \label{line:async-d-dis-2dt:init}
\smallskip

\ProcCode{\ta{$p^i \in P_i, i \in [1..d]$}}{ \label{line:async-d-dis-2dt:Pv}
    $\ooutput$ $v_i$; \label{line:async-d-dis-2dt:out-v} \\
    $\communicate$ $\outputi(v_i)$. \label{line:async-d-dis-2dt:comm}
}
\smallskip

\ProcCode{\ta{$p^? \in P_?$}}{ \label{line:async-d-dis-2dt:pComp}
    $\wait$ until \ta{$p^?$} $\observed$ $d-1$ distinct $\outputi(\star)$; \label{line:async-d-dis-2dt:pComp-wait} \\
    \ta{$V_\obs \gets \{v \mid p^? \text{ \observed\ } \outputi(v)\}$}; \label{line:async-d-dis-2dt:pComp-Vobs} \\
    $\ooutput$ $v \in \{v_1,...,v_d\} \setminus V_\obs$. \label{line:async-d-dis-2dt:pComp-output}
}

\caption{Asynchronous $d$-disagreement algorithm for every SOS \ta{problem $\pb_O$} with $O = \{\{v_1,...,v_d\}\}$, assuming 
\jw{$n \geq d\big\lceil\frac{f+1}{2}\big\rceil+ \big\lfloor\frac{f+1}{2}\big\rfloor$}.}
\label{alg:async-d-dis-2dt}
\end{algorithm}


For every subset $P_i, i \in [1..d]$, a process $p_i \in P_i$ outputs $v_i$ (\zcref{line:async-d-dis-2dt:out-v}) and then communicates $\outputi(v_i)$ (\zcref{line:async-d-dis-2dt:comm}).
Moreover, every process \ta{$p^? \in P_?$} waits until it observes $d-1$ distinct $\outputi(\star)$ (\zcref{line:async-d-dis-2dt:pComp-wait}), and gathers the observed values in the set $V_\obs$ (\zcref{line:async-d-dis-2dt:pComp-Vobs}).
Finally, \ta{$p^?$} outputs the only value that is not present in $V_\obs$ (\zcref{line:async-d-dis-2dt:pComp-output}).

\begin{theorem}[Correctness of \zcref{alg:async-d-dis-2dt}] \label{thm:d-dis-alg}
\zcref{alg:async-d-dis-2dt} implements all $d$-disagreement SOS \ta{problems $\pb_O$} with SOS $O = \{\{v_1,...,v_d\}\}$ under asynchrony and up to \jw{$f$} crash failures, assuming \jw{$n \geq d \big\lceil \frac{f+1}{2} \big\rceil + \big\lfloor \frac{f+1}{2} \big\rfloor$}.
\end{theorem}

\begin{proof}
We prove the safety and completeness of \zcref{alg:async-d-dis-2dt}.

\paragraph{Safety.}
%
\newcommand{\corrP}{\ensuremath{\ta{\mathcal{P}}}\xspace}
Let us consider an arbitrary execution $E$, and let \ta{$\corrP \subseteq \{P_1,...,P_d,P_?\}$ be the set of all subsets of processes of the partition such that each $P' \in \corrP$ contains at least one correct process $p \in P'$ in execution $E$.}
By construction of \zcref{line:async-d-dis-2dt:init}, it is impossible to crash all the processes of two subsets of $P_1,...,P_d,P_?$, so we have \ta{$d \leq |\corrP| \leq d+1$}.
%
We now consider the following two cases.
\begin{enumerate}[leftmargin=*]
    \item Case \textit{(i)}: $P_? \notin \corrP$.
    In this case, we necessarily have $\corrP=\{P_1,...,P_d\}$, thus there exists at least one correct process $p_i$ in every subset $P_i, i \in [1..d]$ that outputs $v_i$ at \zcref{line:async-d-dis-2dt:out-v}.
    Therefore, such an execution necessarily produces the output set $\{v_1,...,v_d\}$.

    \item Case \textit{(ii)}: $P_? \in \corrP$.
    In this case, there can be at most one subset of $P_1,...,P_d$ that does not belong to $\corrP$.
    Let $\corrP' = \corrP \setminus \{P_?\}$ be the set of subsets of processes in $P_1,...,P_d$ that contain correct processes in execution $E$.
    We necessarily have $|\corrP'| \geq d-1$.
    For every subset $P_i \in \corrP'$, there is at least one correct processes \ta{$p^i \in P_i$} that outputs $v_i$ at \zcref{line:async-d-dis-2dt:out-v} and then communicates $\outputi(v_i)$ at \zcref{line:async-d-dis-2dt:comm}.
    By C-Local-Termination and C-Global-Termination, \ta{all} correct processes \ta{$p^? \in P_?$} eventually observe $d-1$ distinct $\outputi(v)$ in \zcref{line:async-d-dis-2dt:pComp-wait} (by case assumption, there is at least one correct process in $P_?$).
    By construction, the set $V_\obs$ of every correct process \ta{$p^? \in P_?$} therefore contains $d-1$ different values at \zcref{line:async-d-dis-2dt:pComp-Vobs} (though the set $V_\obs$ may differ between correct processes of $P_?$).
    For every correct process \ta{$p^? \in P_?$}, $\{v_1,...,v_d\} \setminus V_\obs$ contains exactly one value $v \in \{v_1,...,v_d\}$.
    Finally, every correct process \ta{$p^? \in P_?$} outputs the only value $v \in \{v_1,...,v_d\} \setminus V_\obs$ (\zcref{line:async-d-dis-2dt:pComp-output}) that has not been observed by \ta{$p^?$}.
    Hence, the execution produces the output set $\{v_1,...,v_d\}$.
\end{enumerate}



\paragraph{Completeness.}
Since there is only one output set in $O$, and safety shows that any execution produces this unique output set, then completeness follows immediately.
\end{proof}

\section{Conclusions} \label{sec:conclusion}

In this paper, we introduce and study SOS \ta{problems}, a novel class of distributed \ta{problems} defined by the set of distinct output sets they produce across their executions.
By departing from the classical constraints of simplicial complexes, we have developed a framework that \ta{accommodates output domains that are not closed under inclusion}.

Our main result establishes that the entire class of SOS \ta{problems} is decidable under asynchrony and any number of crashes.
The decision rule is remarkably simple: an SOS \ta{problem} is always solvable when no process may crash, and it is solvable under any positive number of crashes if and only if its SOS graph (whose vertices are the output sets and whose edges connect pairs related by inclusion) is connected.
This clean dichotomy partitions SOS tasks into those that tolerate arbitrarily many crashes and those that tolerate none at all, with no intermediate regime.

Our decidability result has yielded several noteworthy consequences.
First, it reveals a sharp discontinuity in the landscape of $k$-set agreement: without validity, $k$-set agreement is solvable under any number of crashes for $k \geq 2$, yet becomes impossible under even a single crash when $k = 1$ (consensus).
This highlights that the classical impossibility of $k$-set agreement for $k \geq 2$ fundamentally depends on the validity requirement, a dependency that does not exist for consensus, whose impossibility holds even in its SOS form.
Second, our study of $d$-disagreement, a symmetry-breaking \ta{problem} requiring the system to always produce exactly $d$ distinct output values, has uncovered \ta{an interesting} connection to the harmonic series: the number of processes necessary to solve $d$-disagreement under \jw{$f$} crashes is approximately \jw{$n \geq H_d(f+1)$}, where $H_d$ is the $d$-th harmonic number.

Several directions for future work emerge naturally from this study.
First, closing the gap between the lower bound of \jw{$\sum_{i=1}^{d} \ceil{\frac{f+1}{i}}$} and the upper bound of \jw{$d \ceil{\frac{f+1}{2}} + \floor{\frac{f+1}{2}}$} for $d$-disagreement remains an open problem, and an optimal algorithm would be of both theoretical and practical interest, particularly for fault-tolerant load balancing applications.
Second, the universal algorithm we have presented for connected SOS \ta{problems} (\zcref{alg:async-resil}) requires \jw{$n \geq (f+1)|V|$} processes.
It would be valuable to determine tight resilience bounds, that is, the exact relationship between $n$ and \jw{$f$} for which each connected SOS \ta{problem} becomes solvable.
Third, and perhaps most ambitiously, our work suggests a broader program: extending the decidability analysis to richer classes of \ta{problems}, such as \ta{task problems} defined by their set of output vectors (rather than output sets), or \ta{task problem} that incorporate input-dependent specifications.
By progressively mapping the boundary between decidable and undecidable families, we hope to contribute to an eventual comprehensive theory of distributed decidability, one that plays a role for distributed computing analogous to that of classical computability theory for sequential computing.

\ta{
\section*{Acknowledgments}
Supported by grants DRONAC (PID2022-140560OB-I00) and Excelencia María de Maeztu (CEX2024-001471-M) funded by MICIU/AEI/10.13039/501100011033, and by grant PIPF-2024/COM-34475 funded by the Madrid Regional Government.
The authors would like to thank the anonymous reviewers who helped improve this paper.
}




\newpage

\bibliographystyle{IEEEtranS}
\bibliography{bibliography}

\newpage
\appendix




\section{The unsolvability of disconnected SOS tasks under crashes (Necessity)} \label{sec:disconnected}

In this section, we present \zcref{thm:disconnected}, which generalizes the famous impossibility of asynchronous resilient consensus (FLP theorem~\cite{FLP85}) to a broader family of \ta{problems} with non-connected SOS graphs.
Results similar to \zcref{thm:disconnected} have already been shown~\cite{MW87,BMZ90}, however, they rely on model-specific arguments (esp. message passing) and on a reduction from the consensus problem.
Therefore, the FLP theorem is not a consequence of these results, but a cause.
In contrast, our proof of \zcref{thm:disconnected} is agnostic of the underlying communication medium, and the impossibility is proved from scratch, which allows us to state the FLP theorem as a corollary of our theorem (\zcref{cor:flp}).

Our proof generalizes the axiomatic approach of~\cite{AFGGNW24}: asynchrony, crash resilience, and termination are defined as axioms, and the impossibility of implementation is expressed as a contradiction within this system of axioms.
We also generalize the classical notion of \emph{valence}, introduced in~\cite{FLP85}.
However, unlike~\cite{AFGGNW24,FLP85}, we extend the impossibility proof to any disconnected SOS \ta{problem}.


\paragraph{Preliminaries.}
We now introduce the necessary notions for our proof.

\paragraph{Events.}
An \emph{event} $e_p^\idx$ is an action involving process $p$ at index $\idx$ (\ie $e_p^\idx$ is the $i$\textsuperscript{th} event of $p$).

The only purpose of event indices is to differentiate events on the same process that could swap orders due to asynchrony.
For instance, in asynchronous message passing, two messages could be sent by two different senders $p',p''$ to the same recipient $p$, and due to asynchrony, there could be one execution where $p$ first receives $p'$ and then $p''$, and another execution where $p$ first receives $p''$ and then $p'$.
However, thanks to indices, the reception by $p$ of the message from $p'$ in these two executions is considered as two different events, since the indices of these events are different.
Apart from this technical detail, the event indices are not used in the rest of the proof.

Some special events, the input and output events, respectively, correspond to the input/output of a value $v$ by a process $p$ in the context of the task. These events are respectively noted $\iin_p^v$ and $\out_p^v$.
If the process $p$ and/or index $\idx$ of some event are not relevant, we can omit writing them.

\paragraph{States.}
A \emph{state} $\sstate$ is a set of events.
An algorithm $A$ is described by its set of states $\States$ and set of \emph{output states} $\Omega \subseteq \States$, which are the final states where the algorithm can terminate.
By definition, all maximal states of $\States$ are output states, \ie $\{\sstate \in \States \mid \Nexists \sstate' \in \States, \sstate \subsetneq \sstate'\} \subseteq \Omega$, but let us remark that there can also be output states that are subsets of the maximal states.
Conversely, the \emph{input states} are the states from which the algorithm begins and that have no previous state.
They are given by:
$$
\instates(\States) = \{\sstate \in \States \mid \Nexists \sstate' \in \States: \sstate' \subsetneq \sstate\}.
$$

\begin{definition}[Input/output states of \ta{a task implementation}] \label{def:in-out-impl}
If an algorithm with a set of states $\States$ and a set of output states $\Omega$ implements the \ta{task $\task=(\Vvecin,\Vvecout,\map) \in \pb$ of some problem $\pb$}, then there must be input/output states corresponding to every pair of input/output vectors in \ta{$\task$}:
\ta{\begin{align*}
    \Forall (\Vin,\Vout) &\in \map, \Exists \sstate_\iin \in \instates(\States), \Exists \sstate_\out \in \Omega: \sstate_\iin \subseteq \sstate_\out \\
    \land\; \{&\iin_{p_j}^v \mid v \text{ is the } j^\text{th} \text{ value in } \Vin, v \neq \bot\} = \sstate_\iin \\
    \land\; \{&\out_{p_j}^v \mid v \text{ is the } j^\text{th} \text{ value in } \Vout, v \neq \bot\} \subseteq \sstate_\out.
\end{align*}}
\end{definition}

\begin{definition}[SOS valence] \label{def:sos-valence}
Given a set of states $\States$, a set of output states $\Omega$, and a state $\sstate \in \States$, the \emph{SOS valence} of $\sstate$ is a set of sets of output values given by the function:
\begin{align*}
\valence(\sstate,\States,\Omega) =
\begin{cases}
\{\{v \mid \out_p^v \in \sstate\}\} & \text{if } \sstate \in \Omega, \\
\bigcup_{\sstate' \in \{\sstate' \in \States \mid \sstate \subsetneq \sstate'\}} \valence(\sstate',\States,\Omega) & \text{otherwise.}
\end{cases}
\end{align*}
    
\end{definition}

Informally, the SOS valence of an output state is the singleton containing the output set it produces, and the SOS valence of any other state is the union of the SOS valence of all its extensions.
(Let us note that this definition of valence differs from that of~\cite{AFGGNW24,FLP85}, since here SOS valence is the set of reachable \textit{sets} of output values, and not the set of reachable output values, as executions producing multiple output values are allowed in our context.)

\paragraph{Axioms.}
Here, we define asynchrony, resilience, and termination as axioms that some algorithm must satisfy.
In these definitions, and in the rest of the proof, we use the symbol $\uplus$ to denote the union of two disjoint sets.

\begin{definition}[Asynchrony axiom] \label{def:async}
Given a set of states $\States$, \emph{asynchrony} is defined as:
\begin{align*}
    \Asynchrony(\States) &= \Forall \sstate \in \States: \\ (\sstate {\uplus} \{e_p\}, \sstate &{\uplus} \{e_{p'}\} \in \States, p \neq p') {\implies} (\sstate {\uplus} \{e_p, e_{p'}\} \in \States).
\end{align*}
\end{definition}

\Asynchrony (sometimes referred to as the \emph{diamond property}) requires that if two states differ only in their last respective events, which are from different processes, their union is also a state, as the pace of processes can change across two executions.
We point out that our impossibility proof is agnostic of the communication medium object, as long as the medium satisfies asynchrony as defined above (which is the case for the \communicate/\observe abstraction used in this paper).\footnote{
    For a proof that some classical communication media (such as asynchronous message-passing and atomic memory) satisfy this definition of asynchrony, we refer the interested reader to~\cite{AFGGNW24}.
}



\begin{definition}[Termination axiom] \label{def:termin}
Given a set of states $\States$, \emph{termination} is defined as:
$$
\Termination(\States) = \Forall \sstate \in \States: |\sstate| < +\infty.
$$
\end{definition}

\Termination ensures that an algorithm does not have any state with an infinite number of events, and thus that the output states are reached in a finite number of steps.
One could argue that it is allowed for an algorithm to keep making steps indefinitely after reaching an output state, but without loss of generality, we omit this scenario by removing the states extending an output state.

\begin{definition}[Resilience axiom] \label{def:resil}
Given a set of states $\States$ and set of output states $\Omega \subseteq \States$, \emph{resilience} is defined as:
\begin{align*}
    \Resilience(\States,\Omega) &= \Forall \sstate \in \States \setminus \Omega, \\ &\Exists \sstate \uplus \{e_{p}\}, \sstate \uplus \{e_{p'}\} \in \States: p \neq p'.
\end{align*}
\end{definition}

\Resilience imposes that at least two distinct processes can extend every state that is not an output state.
If this is not the case, that is, if a state that is not an output state can only be extended by at most one process, then crashing this process would make the algorithm stop progressing, and thus never reach termination.


\paragraph{Impossibility proof.}
Before proceeding to the impossibility theorem, we first establish some preliminary results.
For notational simplicity, given a set of states $\States$ and a set of output states $\Omega$, we say that a state $\sstate \in \States$ is \emph{disconnected} if its SOS valence has an associated SOS graph $\OutputGraph(\valence(\sstate,\States,\Omega))$ that is disconnected, otherwise we say that state $\sstate$ is \emph{connected}.

\begin{observation} \label{obs:empty-notin-disconnected}
Given a set of states $\States$ and a set of output states $\Omega \subseteq \States$, by \zcref{def:output-graph}, a disconnected $\sstate$ cannot contain the empty set in its SOS valence $\valence(\sstate,\States,\Omega)$, because $\varnothing$ is included in any other set, which would make $\sstate$ connected.
\end{observation}

\begin{lemma} \label{lem:connected-subset-component}
Let $O$ be an SOS whose SOS graph $\OutputGraph(O)$ is disconnected, and let $C_1, \ldots, C_k$ ($k \geq 2$) be its connected components (expressed as sets of vertices).
If $O' \subseteq O$ and $\OutputGraph(O')$ is connected, then $O' \subseteq C_i$ for some $i \in [1..k]$.
\end{lemma}

\begin{proof}
By definition of connected components, there is no edge in $\OutputGraph(O)$ between any vertex in $C_i$ and any vertex in $C_j$ for $i \neq j$.
Since edges in $\OutputGraph(O')$ are a subset of those in $\OutputGraph(O)$, if $O'$ contained vertices from two distinct components $C_i$ and $C_j$, then $\OutputGraph(O')$ would itself be disconnected.
Contradiction.
\end{proof}

\begin{lemma} \label{lem:nonempty-valence}
Given a set of states $\States$ and a set of output states $\Omega \subseteq \States$, all states $\sstate \in \States$ have a nonempty SOS valence: $\valence(\sstate,\States,\Omega) \neq \varnothing$.
\end{lemma}

\begin{proof}
Let us consider some state $\sstate \in \States$.
By definition, all maximal states of $\States$ (\ie states that cannot be extended) are output states: $\{\sstate' \in \States \mid \Nexists \sstate'' \in \States, \sstate' \subsetneq \sstate''\} \subseteq \Omega$.
By \zcref{def:sos-valence}, the SOS valence of an output state $\sstate'$ (and in particular of a maximal state) is the singleton containing all output values of $\sstate'$, and it is therefore not empty.
State $\sstate$ is either an output state, or it is a subset of an output state: $\Exists \sstate' \in \Omega, \sstate \subseteq \sstate'$.
Again, by \zcref{def:sos-valence}, the SOS valence of $\sstate$ must include the SOS valence of $\sstate'$: $\valence(\sstate',\States,\Omega) \subseteq \valence(\sstate,\States,\Omega)$.
Therefore, the SOS valence of $\sstate$ is not empty.
\end{proof}

Note that the previous lemma forbids the existence of an SOS valence of $\varnothing$, but it does not forbid the existence of an SOS valence of $\{\varnothing\}$ (\ie the empty output set may be allowed).

\begin{lemma} \label{lem:discon-valence}
Given a set of states $\States$ and a set of output states $\Omega \subseteq \States$, all disconnected states $\sstate \in \States$ are not output states, \ie $\sstate \notin \Omega$.
\end{lemma}

\begin{proof}
Let us consider some disconnected state $\sstate \in \States$ with SOS valence $O=\valence(\sstate,\States,\Omega)$.
That is, $\OutputGraph(O)$ is disconnected.
By definition of $\OutputGraph()$ (\zcref{def:output-graph}), there must exist at least two output sets $o,o' \in O$ such that $o \cap o' = \varnothing$.
Therefore, $\bigcap O = \varnothing$.
As there is no common value across all the output sets $o \in O$, then there cannot be any output event $\out_p^v$ in $\sstate$, and the set of output values of $\sstate$ is the empty set $\varnothing$.
But by \zcref{obs:empty-notin-disconnected}, as $\sstate$ is disconnected, then $O$ cannot contain the empty output set $\varnothing$, and thus $\varnothing$ is not a valid output set where the algorithm can terminate.
Therefore, $\sstate \notin \Omega$.
\end{proof}

\begin{lemma} \label{lem:disco-state-exists}
\sloppy
An algorithm $A$ (with set of states $\States$ and set of output states $\Omega \subseteq \States$) that solves a disconnected SOS \ta{problem $\pb_O$} must have a disconnected state: $\Exists \sstate \in \States: \OutputGraph(\valence(\sstate,\States,\Omega))$ is a disconnected graph.
\end{lemma}

\begin{proof}
\ta{As algorithm $A$ solves problem $\pb_O$, then it implements some task $\task = (\Vvecin,\Vvecout,\map) \in \pb_O$ (see \zcref{sec:pb-formal}).}
By \zcref{def:in-out-impl}, $\States$ must contain input/output states corresponding to every pair of input/output vectors in \ta{$\task$}.
By the fact that SOS \ta{problems} are validity-free (\zcref{def:sos-pb}), every input vector \ta{$\Vin \in \Vvecin$} can yield all output vectors \ta{$\Vout \in \Vvecout$ (\ie $\map = \Vvecin \times \Vvecout$)}.
This implies that every input state can reach every output set $o \in O$, or more formally: $\Forall \sstate_\iin \in \instates(\States), \Forall o \in O, \Exists \sstate_\out \in \Omega, \sstate_\iin \subseteq \sstate_\out: \{v \mid \iin_p^v \in \sstate_\out\}=o$.
By the definition of SOS valence (\zcref{def:sos-valence}), all input states therefore have the entire SOS $O$ as their SOS valence, \ie $\Forall \sstate_\iin \in \instates: \valence(\sstate_\iin,\States,\Omega)=O$.
By the fact that \ta{$\pb_O$ is a disconnected SOS problem}, $\OutputGraph(O)$ is a disconnected graph, and therefore any input state $\sstate_\iin \in \instates$ is disconnected.
\end{proof}

\noindent
\textbf{Theorem~\ref{thm:disconnected}} (Non-resilience of disconnected SOS \ta{problems})\textbf{.} 
\emph{A disconnected SOS \ta{problem $\pb_O$} can be \ta{solved} asynchronously only if \ta{$f=0$}.}

\begin{proof}
By way of contradiction, let us assume that there is an algorithm $A$ that implements any \ta{task $\task \in \pb_O$} under asynchrony and \jw{$f>0$}, and where \ta{$\pb_O$ is a disconnected SOS problem} with SOS $O$.
Let us consider the set of states $\States$ and the set of output states $\Omega \subseteq \States$ of $A$.
As communication is asynchronous, $\Asynchrony(\States)$ is verified.
Moreover, as \jw{$f>0$}, there can be at least one crash in the system, and $\Resilience(\States,\Omega)$ is verified.
Finally, a task execution must terminate by definition, so $\Termination(\States)$ is verified.

By \zcref{lem:disco-state-exists}, disconnected states must exist.
Let us consider any disconnected state $\sstate_m \in \States$.
We now inductively show that there exists what we call a \emph{critical} state $\sstate_c$, \ie a disconnected state that only has connected extensions: 
\begin{align*}
    \Exists \sstate_c &\in \States: \\ (&\OutputGraph(\valence(\sstate_c,\States,\Omega)) \text{ is a disconnected graph}) \;\land \\
    (&\Forall \sstate' \in \States, \sstate_c \subsetneq \sstate': \\
    &\hspace{2em} \OutputGraph(\valence(\sstate',\States,\Omega)) \text{ is a connected graph}).
\end{align*}
As $\sstate_m$ has a disconnected SOS valence, \zcref{lem:discon-valence} implies that is not an output state, \ie $\sstate_m \notin \Omega$.
By $\Resilience(\States,\Omega)$, $\sstate_m$ must have some extensions by one event.
If all extensions have an connected SOS valence, then $\sstate_m$ satisfies the property of a critical state and we set $\sstate_c=\sstate_m$.
Otherwise, $\sstate_m$ has some extension by one event $\sstate_m \uplus \{e\} \in \States$ that has a disconnected SOS valence.
Then, we make $\sstate_m$ this new extension $\sstate_m \uplus \{e\}$ and repeat this procedure.
Observe that this process must eventually end by finding a critical state, since otherwise, it means that an infinite state exists, which contradicts $\Termination(\States)$.

\sloppy
Let $O_c = \valence(\sstate_c,\States,\Omega)$.
Since state $\sstate_c$ is disconnected, $\OutputGraph(O_c)$ has at least two connected components (expressed as sets of vertices); fix two of them, $C_1$ and $C_2$.
We will now show that there exist two immediate $1$-event extensions $\sstate_1=\sstate_c \uplus \{e_p\}$ and $\sstate_2=\sstate_c \uplus \{e_{p'}\}$ in $\States$ such that $\valence(\sstate_1,\States,\Omega) \subseteq C_1$ and $\valence(\sstate_2,\States,\Omega) \subseteq C_2$ (and in particular $\valence(\sstate_1,\States,\Omega) \cap \valence(\sstate_2,\States,\Omega) = \varnothing$).

By \zcref{def:sos-valence}, $O_c = \bigcup_{\sstate' \in \States,\, \sstate_c \subsetneq \sstate'} \valence(\sstate',\States,\Omega)$.
In particular, $O_c = \bigcup_{\sstate_c \uplus \{e\} \in \States} \valence(\sstate_c \uplus \{e\},\States,\Omega)$, since every extension of $\sstate_c$ passes through some immediate $1$-event extension.
As $O_c$ intersects both $C_1$ and $C_2$, there exist immediate extensions $\sstate_1, \sstate_2$ with $\valence(\sstate_1,\States,\Omega) \cap C_1 \neq \varnothing$ and $\valence(\sstate_2,\States,\Omega) \cap C_2 \neq \varnothing$.
By criticality of $\sstate_c$, both $\valence(\sstate_1,\States,\Omega)$ and $\valence(\sstate_2,\States,\Omega)$ have connected SOS graphs.
Since $\valence(\sstate_1,\States,\Omega), \valence(\sstate_2,\States,\Omega) \subseteq O_c$, \zcref{lem:connected-subset-component} implies that each is included in a single connected component of $\OutputGraph(O_c)$.
Hence $\valence(\sstate_1,\States,\Omega) \subseteq C_1$ and $\valence(\sstate_2,\States,\Omega) \subseteq C_2$, so $\valence(\sstate_1,\States,\Omega) \cap \valence(\sstate_2,\States,\Omega) = \varnothing$.
We now derive a contradiction by considering two cases.
\begin{itemize}[leftmargin=*]
    \item Case 1: $p \neq p'$.
    Given that the processes of the two events are distinct, from $\Asynchrony(\States)$, we have $\sstate''=\sstate_c \uplus \{e_p,e_{p'}\} \in \States$.
    Since $\sstate''$ extends both $\sstate_1$ and $\sstate_2$, \zcref{def:sos-valence} gives $\valence(\sstate'',\States,\Omega) \subseteq \valence(\sstate_1,\States,\Omega) \cap \valence(\sstate_2,\States,\Omega)= \varnothing$.
    Thus $\valence(\sstate'',\States,\Omega)=\varnothing$, which contradicts \zcref{lem:nonempty-valence}.
    
    \item Case 2: $p = p'$.
    By $\Resilience(\States,\Omega)$ (since $\sstate_c \notin \Omega$ by \zcref{lem:discon-valence}), there exists an immediate extension $\sstate_3=\sstate_c \uplus \{e_{p''}\} \in \States$ with $p'' \neq p$.
    By criticality of $\sstate_c$, $\valence(\sstate_3,\States,\Omega)$ has a connected SOS graph.
    Since $\valence(\sstate_3,\States,\Omega) \subseteq O_c$, \zcref{lem:connected-subset-component} implies $\valence(\sstate_3,\States,\Omega) \subseteq C_j$ for some $j$.
    Without loss of generality, suppose $\valence(\sstate_3,\States,\Omega) \cap C_1 = \varnothing$ (if $\valence(\sstate_3,\States,\Omega) \subseteq C_1$, use $C_2$ and $\sstate_2$ instead in the following).
    Then $\valence(\sstate_3,\States,\Omega) \cap \valence(\sstate_1,\States,\Omega) = \varnothing$, since $\valence(\sstate_1,\States,\Omega) \subseteq C_1$.
    As $p'' \neq p$, by $\Asynchrony(\States)$, $\sstate_c \uplus \{e_p, e_{p''}\} \in \States$.
    Since this state extends both $\sstate_1$ and $\sstate_3$, \zcref{def:sos-valence} gives $\valence(\sstate_c \uplus \{e_p,e_{p''}\},\States,\Omega) \subseteq \valence(\sstate_1,\States,\Omega) \cap \valence(\sstate_3,\States,\Omega)= \varnothing$, contradicting \zcref{lem:nonempty-valence}.
    \qedhere
\end{itemize}
\end{proof}

\end{document}